\documentclass[11pt,reqno]{amsart}

\usepackage{mathrsfs}
\usepackage{amsmath}
\usepackage{amssymb}
\usepackage{amsfonts}
\usepackage{amsthm}
\usepackage{a4wide}
\usepackage{graphicx}
\usepackage{color}
\usepackage[sans]{dsfont}
\usepackage{linearA}
\usepackage{pgfkeys}
\usepackage{longtable}
\usepackage{pdflscape}
\usepackage{float}
\usepackage{cancel}
\usepackage{dsfont}
\usepackage{setspace}
\usepackage{array}
\usepackage{hyperref}
\hypersetup{colorlinks=true,allcolors=[rgb]{0,0,0.6}}
\usepackage[english]{babel}
\usepackage[latin1]{inputenc}
\usepackage[T1]{fontenc}
\usepackage{tikz,pgflibraryshapes}
\usepackage{enumitem}
\usepackage{eufrak}
\usepackage[square,numbers,sort&compress]{natbib}
\usepackage[shadow,textwidth=22mm,textsize=tiny]{todonotes}

\newtheorem{theorem}{Theorem}[section]
\newtheorem{lemma}[theorem]{Lemma}
\newtheorem{proposition}[theorem]{Proposition}
\newtheorem{corollary}[theorem]{Corollary}

\newtheorem{remark}[theorem]{Remark}
\newtheorem{assumption}[theorem]{Assumption}

\numberwithin{equation}{section}
\DeclareMathOperator{\slim}{s-lim}

\author[M. Falconi]{Marco Falconi} \address[M. Falconi]{Institut f{\"u}r
  Analysis, Dynamik und Modellierung\\ Universit{\"a}t Stuttgart;
  Pfaffenwaldring 57 D-70569 Stuttgart, Deutschland}
\email{marco.falconi@mathematik.uni-stuttgart.de}
\urladdr{http://www.mathematik.uni-stuttgart.de/~falconmo/}

\author[J. Faupin]{J{\'e}r{\'e}my Faupin}
\address[J. Faupin]{Institut Elie Cartan de Lorraine \\
Universit{\'e} de Lorraine, 
57045 Metz Cedex 1, France}
\email{jeremy.faupin@univ-lorraine.fr}

\author[J. Fr\"ohlich]{J\"urg Fr\"ohlich}
\address[J. Fr{\"o}hlich]{Institut f{\"u}r Theoretische Physik, ETH H{\"o}nggerberg, CH-8093 Z{\"u}rich, Switzerland}
\email{juerg@phys.ethz.ch}

\author[B. Schubnel]{Baptiste Schubnel}
\address[B. Schubnel]{Institut Elie Cartan de Lorraine \\
Universit{\'e} de Lorraine, 
57045 Metz Cedex 1, France \\
\newline and 
SBB Personenverkehr,  Wylerstrasse 125, 3014 Bern, Switzerland}
\email{baptiste.schubnel@yahoo.fr}

\begin{document}
\bibliographystyle{abbrv} \title[Scattering for
Lindbladians]{Scattering theory for Lindblad master equations}

\begin{abstract}
We study scattering theory for a quantum-mechanical system consisting of a particle scattered off a dynamical target that occupies a compact region in position space. After taking a trace over the degrees of freedom of the target, the dynamics of the particle is generated by a Lindbladian acting on the space of trace-class operators. We study scattering theory for a general class of Lindbladians with bounded interaction terms. First, we consider models where a particle approaching the target is always re-emitted by the target. Then we study models where the particle may be captured by the target. An important ingredient of our analysis is a scattering theory for dissipative operators on Hilbert space.
\end{abstract}

\maketitle

\section{Introduction and statement of the main results}

We study the quantum-mechanical scattering theory for particles interacting with a dynamical target. The target may be a quantum field, e.g., a phonon field of a crystal lattice, a quantum gas, or a solid, such as a ferro-magnet, \dots confined to a compact region of physical space $\mathbb{R}^{3}$. Our aim in this paper is to contribute to a mathematically rigorous description of such scattering processes and to provide a mathematical analysis of particle capture by the target. Rather than studying all the degrees of freedom of the total system composed of particles and target, we will take a trace over the degrees of freedom of the target and study the reduced (effective) dynamics of the particles. It is known that, in the kinetic limit (time, $t$, of order 
$\lambda^{-2}$, with $\lambda \rightarrow 0$, where $\lambda$ is the strength of interactions between the particles and the target), the reduced dynamics of the particles is \textit{not} unitary, but is given by a semi-group of completely positive operators generated by a Lindblad operator. In general, the reduced time evolution maps pure states to mixed states corresponding to density matrices. The trace of a density matrix tends to decrease under the reduced time evolution; but, in the absence of particle capture by the target, it is preserved. 

The main purpose of this paper is to study the dynamics generated by general Lindblad operators and, in particular, to develop the scattering theory for Lindblad operators. We will also study models of some concrete physical systems.

In the remainder of this section, we recall the definition of Lindblad operators and quantum dynamical semigroups (see \cite{AlickiLendi} for a detailed introduction to the subject), we discuss general features of the scattering theory for Lindblad master equations and we state our main results.

\subsection{Lindblad operators and quantum dynamical semigroups}\label{S2}

To avoid inessential technicalities, we cast our analysis in the language of
operators on Hilbert-space; but our discussion can easily be generalized using the language of operator algebras.

Thus, let $\mathcal{H}$ be the complex separable Hilbert space of state vectors of an open quantum-mechanical system $S$. We will use the Schr\"{o}dinger picture to describe the time evolution of $S$, i.e., the time evolution of normal states of $S$ will be considered. But, as usual, it is possible to reformulate most of the results presented below in the
  Heisenberg picture.
By $\mathcal{J}_1(\mathcal{H})$ and
$\mathcal{J}_{1}^{\mathrm{sa}}(\mathcal{H})$ we denote the complex Banach space of trace-class operators on $\mathcal{H}$ and the real Banach space of \textit{self-adjoint} trace-class operators on
$\mathcal{H}$, respectively. Density matrices, i.e., positive trace-class operators of trace $1$, belong to the cone $\mathcal{J}_{1}^{+}(\mathcal{H}) \subset \mathcal{J}_{1}^{\mathrm{sa}}(\mathcal{H})$. The trace norm in $\mathcal{J}_1( \mathcal{H} )$ is denoted by $\| \cdot \|_1$.

In the kinetic limit (i.e., the Markovian approximation), the time evolution of states of an open quantum system is given by a strongly continuous one-parameter semigroup of trace-preserving and positivity-preserving contractions, $\{T(t)\}_{t \geq 0}$, on $\mathcal{J}_{1}^{\mathrm{sa}}(\mathcal{H})$. 
We remind the reader of the definition and the properties of a strongly continuous semigroup $\{T(t)\}_{t \geq 0}$ on a Banach
space $\mathcal{J}$, (see, e.g., \cite{Davies,Engel}): \vspace{2mm}
\begin{enumerate}
\item\label{item:1} $T(t+s)=T(t) T(s)=T(s) T(t), \quad T(0)= \mathds{1}, \qquad \forall \text{  }t,s \geq 0,  \hfill (semigroup \text{   }property)$ \vspace{3mm}%[4pt]
\item\label{item:2} $t   \mapsto T(t) \rho \text{ is continuous, for all } \rho \in \mathcal{J}.   \hfill (strong \text{ } continuity)$  \vspace{2mm} %[4pt]
\end{enumerate}
If, in addition to (1) and (2), $\{T(t)\}_{t \geq 0}$ also satisfies \vspace{2mm}
\begin{enumerate}
\setcounter{enumi}{2}
\item\label{item:3}
  $\text{ } \|T(t) \rho \| \leq \|\rho \|, \text{ for all }\rho \in
  \mathcal{J},$\hfill\emph{(contractivity)} \vspace{2mm}
\end{enumerate}
then it is called a strongly continuous contraction semigroup.  To qualify as a dynamical map on $\mathcal{J}_{1}^{\mathrm{sa}}(\mathcal{H})$, $\{T(t) \}_{t \geq 0}$ must also preserve positivity and the trace of $\rho$, i.e.,  it must map density matrices to density matrices: \vspace{2mm}
\begin{enumerate}
  \setcounter{enumi}{3}
\item\label{item:4} $\text{ } T(t) \rho \geq 0$, for all $t \geq 0$ and all $\rho \geq 0$, \vspace{3mm} %[4pt]
\item\label{item:5} $\text{ } \mathrm{Tr}(T(t) \rho)= \mathrm{Tr}(\rho)$, for all $\rho \in \mathcal{J}_{1}^{\mathrm{sa}}(\mathcal{H})$.
\end{enumerate}
\vspace{2mm}
\noindent 
In this paper, the generator, $L$, of a strongly continuous semigroup $\{T(t) \}_{t \geq 0}$ on $\mathcal{J}_1(\mathcal{H})$ is defined by
\begin{equation*}
L \rho := \lim_{t\to 0} (-it)^{-1} ( T_t \rho - \rho ) ,
\end{equation*}
the domain of $L$ being the set of trace-class operators $\rho$ such that the limit $t \rightarrow 0$ exists. This is not the usual convention but is natural in our context. We then write $T(t) \equiv e^{ - i t L }$, for all 
$t \ge 0$.

In \cite{kossa} (see also \cite{IngardenKossakowski}) it is shown that necessary and sufficient conditions for a linear operator $L$ on $\mathcal{J}_{1}^{\mathrm{sa}}(\mathcal{H})$ to be the generator of a strongly continuous one-parameter semigroup of trace-preserving and positivity-preserving contractions are that: (i) 
$\mathcal{D}(L)$ is dense in $\mathcal{J}_{1}^{\mathrm{sa}}(\mathcal{H})$, (ii) $\mathrm{Ran} ( \mathrm{Id} - i L ) = \mathcal{J}_{1}^{\mathrm{sa}}(\mathcal{H})$, (iii) $-i\mathrm{Tr}( \mathrm{sgn}( \rho ) L \rho ) \le 0$, for all $\rho \in \mathcal{J}_{1}^{\mathrm{sa}}(\mathcal{H})$, and (iv) 
$\mathrm{Tr}( L \rho ) = 0$, for all $\rho \in \mathcal{J}_{1}^{\mathrm{sa}}(\mathcal{H})$. 

In \cite{MR0413878}, \textit{norm-continuous} semigroups of \textit{completely positive} maps on the algebra (of ``observables'') $\mathcal{B}(\mathcal{H})$ (Heisenberg picture) were studied. We recall that a map $\Lambda$ on $\mathcal{B}( \mathcal{H} )$ is called completely positive iff, for any $n \in \mathbb{N}$, the map $\Lambda \otimes \mathrm{Id}$ on $\mathcal{B}( \mathcal{H} \otimes \mathbb{C}^n)$ is positive. The explicit form of the generators of norm-continuous semigroups of completely positive maps on 
$\mathcal{B}( \mathcal{H} )$ has been found in \cite{MR0413878}. They are called Lindblad generators, or Lindbladians. Translated to the Schr{\"o}dinger picture, which we use in this paper, the results in \cite{MR0413878} imply that Lindblad generators on $\mathcal{J}_{1}^{\mathrm{sa}}(\mathcal{H})$ have the form
\begin{equation} \label{Lili0}
  \mathcal{L} = \mathrm{ad}(H_0)-\frac{i}{2}\sum_{j\in J}
  \{C_j^{*}C_j,\,\cdot \,\}+i\sum_{j\in J}^{}C_j\,\cdot \,C_j^{*},
\end{equation}
where $H_0$ is (bounded and) self-adjoint,
\begin{equation*}
\mathrm{ad}(H_0) := [ H_0 , \cdot ] ,
\end{equation*}
and the operators $C_j$ and $\sum_{j \in J } C_j^* C_j$ are bounded. The operator $\mathcal{L}$ is called a \emph{Lindblad operator} even if some of the operators $H_0$
and/or $C_j$ are unbounded; (we recall that $\mathcal{L}$ generates a norm-continuous semigroup if and only if $\mathcal{L}$ is bounded; see e.g. \cite{Davies}). Strongly continuous one-parameter semigroups of trace-preserving and completely positive  contractions on $\mathcal{J}_1^{\mathrm{sa}}( \mathcal{H} )$ are sometimes called \emph{quantum dynamical semigroups}.

A proof of the following lemma can be found, for instance, in \cite{Daviesbook}. For the convenience of the reader, a proof is reported in Appendix \ref{app:21}.

\begin{lemma} \label{21} Let $H_0$ be a self-adjoint operator on
  $\mathcal{H}$, and let $C_j \in \mathcal{B}(\mathcal{H})$ for all
  $j \in J$ be such that $\sum_{ j \in J } C^*_j C_j \in \mathcal{B}( \mathcal{H} )$. Then the operator $\mathcal{L}$ in Eq. \eqref{Lili0}, with
  domain given by
  \begin{align*}
  \mathcal{D}(\mathcal{L})=& \mathcal{D}(\mathrm{ad}(H_0)) = \big \{ \rho \in \mathcal{J}_{1}(\mathcal{H}) , \rho ( \mathcal{D}( H_0 ) ) \subset \mathcal{D}( H_0 )  \text{ and } \\
  &H_0 \rho - \rho H_0 \text{ defined on } \mathcal{D}( H_0 ) \text{ extends to an element of } \mathcal{J}_{1}(\mathcal{H}) \big \},
  \end{align*}
   is closed and generates a strongly continuous one-parameter semigroup $\{ e^{ - i t \mathcal{L} } \}_{ t \ge 0 }$ 
  on $\mathcal{J}_{1}(\mathcal{H})$ which satisfies properties
  \eqref{item:1}-\eqref{item:2} and \eqref{item:4}-\eqref{item:5}. Moreover for all $t \ge 0$, $e^{ - i t \mathcal{L} }$ is completely positive, and the
  restriction of $\{ e^{ - i t \mathcal{L} } \}_{ t \ge 0 }$ to the Banach space
  $\mathcal{J}_{1}^{\mathrm{sa}}(\mathcal{H})$ is a semigroup of
  contractions, i.e., satisfies \eqref{item:3}.
\end{lemma}

\begin{remark}\label{rk:intro1}
  Since $\|e^{-it \mathcal{L}} \rho \|_1 \leq \|\rho\|_1$ for all
  $\rho \in \mathcal{J}_{1}^{\mathrm{sa}}(\mathcal{H})$, we deduce
  that $\|e^{-it \mathcal{L}} \rho \|_1 \leq 2 \|\rho\|_1$, for all
  $\rho \in \mathcal{J}_{1}(\mathcal{H})$, by using the decomposition
  $\rho= (\rho+ \rho^*)/2 -i (i(\rho-\rho^*))/2$.
\end{remark}

Under some further assumptions, it is possible to treat Lindblad generators with
operators $C_j$'s that are unbounded \cite{davies77,davies79}. However, to avoid inessential technicalities, we will restrict our attention to examples of Lindbladians for which \textit{all}
the operators $C_j$'s are bounded.

\subsection{Wave operators and asymptotic completeness}
\label{sec:wave-oper-asympt}

Next, we discuss some basic concepts in the scattering theory of general semigroups
of operators acting on the Banach space $\mathcal{J}_{1}(\mathcal{H})$. These
concepts can be used in the study of asymptotic behavior of both
Lindblad evolutions and Hilbert-space semigroups. In this section, we do not consider the possibility of ``particle capture'' by a target. But this will be done in Section \ref{subsec:capture}, below.

We suppose that we are given a strongly continuous, uniformly bounded one-parameter
\emph{semigroup} $\{e^{-i t \mathcal{L}}\}_{t\geq 0}$ on $\mathcal{J}_{1}(\mathcal{H})$ and a strongly
continuous \emph{group} $\{e^{-i t \mathcal{L}_0}\}_{t\in \mathbb{R}}$ on 
$\mathcal{J}_{1}(\mathcal{H})$ given by conjugation with unitary operators.
The group $\{e^{-i t \mathcal{L}_0}\}_{t\in \mathbb{R}}$ describes the free dynamics of a particle, while 
$\{e^{-i t \mathcal{L}}\}_{t\geq 0}$ describes the dynamics of a particle interacting with a dynamical target in the Markovian approximation. To simplify matters, we assume that $\mathcal{L}_0$ does not have any eigenvalue.

We are interested in studying
asymptotics of the evolution of the particle state, as $t \to +\infty$. As usual, the guiding idea is that, for large times, one can compare the evolution of a given state
$\rho$ in the presence of interactions with a target with the free evolution of another state, $\rho_0$,
the \textit{scattering state}. As in the standard Hilbert space theory, we cannot compare the two dynamics if we choose an eigenvector of  $\mathcal{L}$ as our initial condition.
It is convenient to assume that the Banach space  $\mathcal{J}_1(\mathcal{H})$ can be decomposed as follows:
\begin{align}
  \label{eq:dp}
  \mathcal{J}_1(\mathcal{H})&=\mathcal{D}\oplus \mathcal{D}_{ \mathrm{pp} } ,
\end{align}
where $\mathcal{D}_{ \mathrm{pp} }$ is the closure of the vector space spanned by all the eigenvectors 
of $\mathcal{L}$ in $\mathcal{J}_1(\mathcal{H})$, and $\mathcal{D}$ is a closed subspace 
complementary to $\mathcal{D}_{ \mathrm{pp} }$. 

Dealing with semigroups $\lbrace e^{-it \mathcal{L}} \rbrace_{t\geq 0}$, the fact that time $t$ has to be taken to be positive makes the analysis of scattering somewhat more subtle. It leads us to define the following two wave operators:
\begin{align}
  \Omega ^+(\mathcal{L},\mathcal{L}_0)&:=\underset{t\to +\infty }{\slim} \:  e^{-it \mathcal{L}}e^{it \mathcal{L}_0}  , \label{eq:defOmega_1} \\
  \Omega ^-(\mathcal{L}_0,\mathcal{L})&:= \underset{t\to +\infty }{\slim} \: e^{it \mathcal{L}_0}e^{-it\mathcal{L}}\Bigl\lvert_{\mathcal{D}}  \; . \label{eq:defOmega_2}
\end{align}
Proving the existence of $ \Omega^+(\mathcal{L},\mathcal{L}_0)$ and
$ \Omega^-(\mathcal{L}_0,\mathcal{L})$ for concrete examples of
Lindblad evolutions is the main purpose of this paper achieved in subsequent sections. 
In the rest of this subsection we assume that \eqref{eq:dp} is valid and that the wave operators
$\Omega^-(\mathcal{L}_0,\mathcal{L})|_{\mathcal{D}}$ and
$\Omega^+(\mathcal{L},\mathcal{L}_{0})|_{\mathcal{J}_1(\mathcal{H})}$ exist. 
For the concrete examples discussed in the following, only the special case
$\mathcal{D} = \mathcal{J}_1(\mathcal{H})$ is relevant.

Let us denote by $\rho\in \mathcal{D}$ an
initial condition (an ``interacting'' vector) for the full time evolution 
$e^{-it\mathcal{L}}$, and by $\rho^+\in \mathcal{J}_1(\mathcal{H})$ an initial condition (``scattering vector'')
for the free evolution $e^{-it\mathcal{L}_{0}}$. 
One of the main goals of scattering theory is to
prove the following convergence:
 For an arbitrary interacting vector $\rho \in \mathcal{D}$, there exists a scattering vector $\rho^{+} \in \mathcal{J}_1(\mathcal{H})$ such that
\begin{equation}
\label{eqq}
\lim_{t\to +\infty}\bigl\lVert e^{-it\mathcal{L}} \rho - e^{-it\mathcal{L}_0}\rho^{+}\bigr\rVert_{1}=0 .
\end{equation}
If \eqref{eqq} is satisfied, we say that $\rho^+$ is the (future) \emph{asymptotic approximation of $\rho$}. The convergence in Eq. \eqref{eqq} is equivalent
to the existence of the wave operator $\Omega^-(\mathcal{L}_0,\mathcal{L})$ on the subspace 
$\mathcal{D}$ and can be
seen as a weak form of \emph{``asymptotic completeness''}. Indeed, the
existence of $\Omega^-(\mathcal{L}_0,\mathcal{L})|_{\mathcal{D}}$ tells us that, to any state 
$\rho\in \mathcal{D}$, (i.e., any state in a subspace complementary to the bound states of 
$\mathcal{L}$), a unique scattering state $\rho^+=\Omega^-( \mathcal{L}_0 , \mathcal{L} ) \rho$
can be associated with the property that \eqref{eqq} holds. This notion of asymptotic
completeness can and ought to be strengthened, as is usually done in standard
quantum mechanical scattering theory on Hilbert space. One natural additional condition strengthening 
\eqref{eqq} is to require that
$\mathrm{Ran}\bigl( \Omega^+(\mathcal{L},\mathcal{L}_0)\bigr)\supseteq
\mathcal{D}$;
i.e., that any $\rho\in \mathcal{D}$ can be written as
$\rho=\Omega^+( \mathcal{L} , \mathcal{L}_0 ) \rho^-$, for a state
$\rho^- \in \mathcal{J}_1( \mathcal{H} )$ (also called a ``scattering state''). A stronger version is 
to require that
$\mathrm{Ran}\bigl( \Omega^+(\mathcal{L},\mathcal{L}_0)\bigr)=
\mathcal{D}$,
which ensures the existence of the ``scattering endomorphism'',
$\sigma:\mathcal{J}_1( \mathcal{H} ) \to
\mathrm{Ran}\bigl( \Omega^-(\mathcal{L}_0,\mathcal{L})\bigr)$,
defined as
\begin{equation}
  \sigma= \Omega ^-(\mathcal{L}_0,\mathcal{L})  \Omega ^+(\mathcal{L},\mathcal{L}_0) .
\end{equation}
If, in addition, $\mathrm{Ran}\bigl(\Omega^-(\mathcal{L}_0,\mathcal{L})\bigr)=
\mathcal{J}_1( \mathcal{H} )$,
then $\sigma:\mathcal{J}_1( \mathcal{H} ) \to \mathcal{J}_1( \mathcal{H} )$ is an invertible endomorphism, i.e., an isomorphism. We say that the wave operators $ \Omega^-(\mathcal{L}_0,\mathcal{L})$ and $ \Omega^+(\mathcal{L},\mathcal{L}_0)$ are \emph{(asymptotically) complete} iff 
\begin{equation*}
\mathrm{Ran}\bigl( \Omega^-(\mathcal{L}_0,\mathcal{L})\bigr) = \mathcal{J}_1( \mathcal{H} ) \quad \text{and    }\quad \mathrm{Ran}\bigl( \Omega^+(\mathcal{L},\mathcal{L}_0)\bigr) = \mathcal{D}.
\end{equation*}

Let $\Omega^-(\mathcal{L}_0,\mathcal{L})^*$ be the adjoint of $\Omega^-(\mathcal{L}_0,\mathcal{L})$ acting on the dual space $\mathcal{B}( \mathcal{H} ) = \mathcal{J}_1( \mathcal{H} )^*$. Given $\rho \in \mathcal{J}_1( \mathcal{H} ) \subset \mathcal{B}( \mathcal{H} )$, the state $\rho^- = \Omega^-( \mathcal{L}_0 , \mathcal{L} )^* \rho$ is a past asymptotic approximation of $\rho$. Choosing, for instance, 
$\rho^- = \rho_{ \mathrm{in} } = |Ê\varphi_{ \mathrm{in}Ê} \rangle \langle \varphi_{ \mathrm{in} } |$, with $\|Ê\varphi_{ \mathrm{in} } \|_{ \mathcal{H} }Ê= 1$, and $\rho^+ = \rho_{ \mathrm{out} } = |Ê\varphi_{ \mathrm{out}Ê} \rangle \langle \varphi_{ \mathrm{out} } | \in \mathcal{D}$, with $\|Ê\varphi_{ \mathrm{out} } \|_{ \mathcal{H} }Ê= 1$, the scattering endomorphism $\sigma$ allows one to compute the transition probability
\begin{align*}
\big \langle \Omega^-( \mathcal{L}_0 , \mathcal{L} )^* \rho_{ \mathrm{out} } , \Omega^+( \mathcal{L} , \mathcal{L}_0 ) \rho_{ \mathrm{in} } \big \rangle_{ ( \mathcal{B}( \mathcal{H} ) ; \mathcal{J}_1( \mathcal{H} ) ) } &= \big \langle \varphi_{ \mathrm{out}Ê} , ( \sigma \rho_{Ê\mathrm{in} }) \varphi_{ \mathrm{out} } \big \rangle_{Ê\mathcal{H} } \\
& = \underset{t\to +\infty }{\slim} \langle \varphi_{ \mathrm{out}Ê} , ( e^{it \mathcal{L}_0} e^{-2 i t \mathcal{L}} e^{ i t \mathcal{L}_0 } \rho_{Ê\mathrm{in} } ) \varphi_{ \mathrm{out}Ê} \big \rangle_{Ê\mathcal{H} }.
\end{align*}
The concepts introduced here are illustrated in the figure below.  
\begin{center}
 \begin{tikzpicture}
\label{fig01}
      
      \fill[black!15]  (5,0) circle (2cm); 
      
        \fill[black!45]  (4.5,0.5) circle (0.5cm); 
         \fill[black!45]  (1.5,-1) circle (0.5cm); 
          \fill[black!45]  (6.5,2.75) circle (0.5cm); 
          
           \fill[black!25]  (5.5,-2.8) circle (0.5cm); 
          \fill[black!25]  (7.3,-2.5) circle (0.5cm); 
          
          \draw[->](1.5,-1) to[bend right=30]   (4.5,0.5);  
 \draw[->] (4.55,0.6) to[bend right=30]   (6.5,2.75);

         \draw[->,dashed]  (5.5,-2.8) to[bend left=10]   (1.55,-1.2);  
         \draw[->,dashed]  (7.3,-2.5) to[bend right=10]   (6.65,2.7); 
          \draw[->,thick,dotted]  (5.65,-2.8) to[bend left=5]   (7.2,-2.4);  
          \draw  (5.5,-1.5) node[above] { \small target};

 \draw (1,-1) node[left] { \small  $\rho^-(t)$};
     \draw (4,0.5) node[left] { \small  $\rho$};    
    \draw (6,2.75) node[left] { \small  $\rho^+(s)$};    
     \draw(6,-3) node[right] { \small  $\rho^{-}$}; 
      \draw(7.8,-2.5) node[right] { \small  $\rho^{+}$};

     \draw(3.5,-2.2) node[below] { \small  $e^{it \mathcal{L}_0}$};  
     \draw(3.2,-0.5) node[above] { \small  $e^{-it \mathcal{L}}$};  
      \draw(5.4,1.3) node[above] { \small  $e^{-is \mathcal{L}}$};  
       \draw(7.6,0) node[right] { \small  $e^{-is \mathcal{L}_0}$};  
         \draw(6.7,-2) node[below] { \small  $\sigma$};

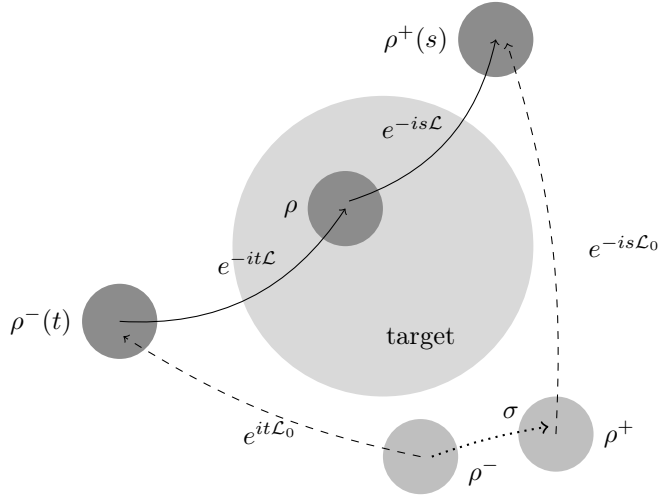
               \draw (5,-4) node[below] { \small Figure 1. Illustration of the scattering operators ($s,t$ must go to $+\infty$)};

        \end{tikzpicture}

\end{center}

\subsection{Statement of the main result}

To avoid cumbersome notations we consider a Lindblad generator given by
\begin{equation}\label{eq:def0_L}
  \mathcal{L} = \mathrm{ad}(H_0) - \frac{i}{2}  \{C^{*}C,\,\cdot \,\} + i C \,\cdot \,C^{*},
\end{equation}
where $H_0$ is a self-adjoint operator on $\mathcal{H}$, and $C \in \mathcal{B}( H )$ is a bounded operator. The analysis of general Lindblad generators, as given in \eqref{Lili0}, can be inferred from 
the one we present in the following by adapting Assumption \ref{Z0}, below. We choose
\begin{equation}
  \mathcal{L}_0 := \mathrm{ad}(H_0).
\end{equation}
Noting that 
\begin{align*}
  \mathcal{L} = H \cdot \,  - \, \cdot H^* + i C \, \cdot \, C^{*} ,
\end{align*}
where 
\begin{equation}\label{eq:def0_H}
H := H_0 - \frac{i}{2} C^* C 
\end{equation}
is a dissipative operator acting on the Hilbert space $\mathcal{H}$, it is useful to compare the semigroup $e^{ - i t \mathcal{L} }$ to the auxiliary semigroup $e^{ - i t H } ( \cdot ) e^{ i t H^* }$. In our analysis, an important role will be played by the operator $H$.

Next, we present the main hypotheses underlying our analysis.
\begin{assumption}\label{Y0}
There exists a dense subset $\mathcal{E} \subset \mathcal{H}$ such that, for all $u \in \mathcal{E}$,
\begin{align}\label{eq:Y0}
\int_{ \mathbb{R} } \big \|ÊC^* C e^{ - i t H_0 } u \big \|_{ \mathcal{H} } dt < \infty.
\end{align}
\end{assumption}
Assumption \ref{Y0} is used to study the scattering theory for the operators $H$ and $H_0$. But we will see that this assumption is also useful in the study of the scattering theory for the Lindblad operators 
$\mathcal{L}$ and $\mathcal{L}_0$.
\begin{assumption}\label{Z0}
There exists a positive constant $\mathrm{c}_0$ depending on $C$ and $H_0$ such that,
\begin{align}\label{eq:Z0}
\int_{ \mathbb{R} } \big \|ÊC e^{ - i t H_0 } u \big \|_{ \mathcal{H} }^2 dt \le \mathrm{c}_0^2 \| u \|_{ \mathcal{H} }^2 ,
\end{align}
for all $u \in \mathcal{H}$.
\end{assumption}
Assumption \ref{Z0} amounts to assuming that the operator $C$ is $H_0$-smooth in the sense of Kato \cite{Kato1}. We recall from \cite{Kato1} that this assumption is equivalent to the inequality
\begin{align}
\int_{ \mathbb{R} } \Big ( \big \|ÊC ( H_0 - ( \lambda + i 0^+ ) )^{-1} u \big \|_{ \mathcal{H} }^2 + \big \|ÊC ( H_0 - ( \lambda - i 0^+ ) )^{-1} u \big \|_{ \mathcal{H} }^2 \Big ) d \lambda \le (\mathrm{c}'_0)^2 \| u \|_{ \mathcal{H} }^2 ,
\end{align}
for some $\mathrm{c}'_0 > 0$ (that can be chosen to be $\mathrm{c}'_0 = 2 \pi \mathrm{c}_0$), which is also equivalent to assuming that
\begin{equation}
\sup_{ z \in \mathbb{C} \setminus \mathbb{R} } \big \| C \big ( ( H_0 - z )^{-1} - ( H_0 - \bar z )^{-1} \big ) C^* \big \|_{ \mathcal{H} } \le \mathrm{c}'_0.
\end{equation}
For other conditions equivalent to \eqref{eq:Z0} we refer to \cite{Kato1}. Obviously, if $u \neq 0$ is an eigenvector of $H_0$, \eqref{eq:Z0} implies that $C u = 0$. In particular, if $\mathrm{Ker}( C ) = \{ 0 \}$ the pure point spectrum of $H_0$ must be assumed to be empty.

We remark that the following bound is \emph{always} satisfied:
\begin{align}\label{eq:Z'''0}
\int_0^\infty \big \|ÊC e^{ - i t H } u \big \|_{ \mathcal{H} }^2 dt \le \| u \|_{ \mathcal{H} }^2 ,
\end{align}
for all $u \in \mathcal{H}$. This follows from the identity
\begin{align}\label{eq:Z'0}
\int_0^t \big \langle u , e^{ i s H^* } C^*ÊC e^{ - i s H } u \rangle ds = - \int_0^t \partial_s \big \langle u , e^{ i s H^* } e^{ - i s H } u \rangle ds = \| u \|_{ \mathcal{H} }^2 - \big \|Êe^{ - i t H }Êu \big \|^2_{ \mathcal{H} } .
\end{align}
Similarly as in \eqref{eq:Z0}, we denote by $\tilde{\mathrm{c}}_0$ the smallest positive constant ($0 < \tilde{\mathrm{c}}_0 \le 1$) with the property that
\begin{align}\label{eq:Z''0}
\int_0^\infty \big \|ÊC e^{ - i t H } u \big \|_{ \mathcal{H} }^2 dt \le \tilde{\mathrm{c}}_0^2 \| u \|_{ \mathcal{H} }^2 ,
\end{align}
for all $u \in \mathcal{H}$. 

One of the main results of this paper is described in the following theorem.
\begin{theorem}\label{thm:main0}
Suppose that either Assumption \ref{Y0} holds, or that Assumption \ref{Z0} holds with $\mathrm{c}_0 < 2$. Then
\begin{equation*}
\Omega^+( \mathcal{L} , \mathcal{L}_0 ) \text{ exists on }  \mathcal{J}_1( \mathcal{H} ).
\end{equation*}
Suppose that Assumption \ref{Z0} holds with $\mathrm{c}_0 < 2$. Then
\begin{equation*}
\Omega^-( \mathcal{L}_0 , \mathcal{L} ) \text{ exists on }  \mathcal{J}_1( \mathcal{H} ).
\end{equation*}
Suppose that Assumption \ref{Z0} holds with $\mathrm{c}_0 < 2 - \sqrt{2}$. Then the wave operators exist and are (asymptotically) complete in the sense of the previous subsection. More precisely, if $\mathrm{c}_0 < 2 - \sqrt{2}$, then $\Omega^+( \mathcal{L} , \mathcal{L}_0 )$ and $\Omega^-( \mathcal{L}_0 , \mathcal{L} )$ are invertible in $\mathcal{B}( \mathcal{J}_1( \mathcal{H} ) )$, and the Lindblad generators $\mathcal{L}$ and $\mathcal{L}_0$ are similar.
\end{theorem}
\begin{remark}
$\quad$
\begin{enumerate}
\item We will prove in Section \ref{sec:dissipative} that Assumption \ref{Z0} with $\mathrm{c}_0 < 2$ implies that \eqref{eq:Z''0} holds with $\tilde{\mathrm{c}}_0 < 1$. 
 It will appear in our proof that sufficient conditions for the existence of the wave operators are that Assumption \ref{Z0} holds and that \eqref{eq:Z''0} holds with $\tilde{Ê\mathrm{c} }_0 < 1$. Furthermore, we will see that the upper bound $\mathrm{c}_0 < 2 - \sqrt{2}$ implies that $\tilde{\mathrm{c}}_0 < 1/ \sqrt{2}$, and sufficient conditions for the completeness of the wave operators are that Assumption \ref{Z0} holds and that \eqref{eq:Z''0} holds with $\tilde{Ê\mathrm{c} }_0 < 1/\sqrt{2}$.
\item We will verify Assumptions \ref{Y0} and \ref{Z0} in some concrete, physically interesting examples, using the explicit form of $e^{ - i t \mathcal{L}_0}$; see Section \ref{sec:example}.
\item To obtain an estimate on $\tilde{\mathrm{c}}_0$ in \eqref{eq:Z''0}, it is possible to apply Mourre's theory for dissipative operators, as developed in \cite{Gol,Roy}. We do, however, not know any examples where an estimate on $\tilde{\mathrm{c}}_0$ obtained with the help of Mourre's theory is better than the one we will obtain in our approach, using perturbative arguments.
\end{enumerate}
\end{remark}

\subsection{Physical context}
\label{sec:intertw-prop-wave}
The abstract notions and concepts formulated above are well-suited to study the
large-time dynamics in interesting models of systems of particles, such as electrons or neutrons, interacting with the degrees of freedom of a dynamical target, which is usually a system of condensed matter, such as an insulator, a metal, or a magnetic material, etc. In these models, the degrees of freedom of the target are ``traced out'', so that time evolution of the particles is not given by a group of unitary
transformations but is assumed to be given by a contraction semi-group of completely positive maps, 
as discussed above, and pure states may thus evolve into mixtures. A concrete example of a physical system that we are able to analyze consists of a beam of independent, spin-polarized electrons transmitted through a magnetized film, as studied in experiments carried out in the group of the 
late H. Chr. Siegmann; see, e.g., \cite{Sieg,Albert-Fr}. In these experiments, the film consists of Iron or Nickel, which are ferromagnetic metals, and exhibits a spontaneous magnetization, $\vec{M}$.
 
\begin{center}
 \begin{tikzpicture}
\label{fig00}
      
      \fill[black!15]  (5,0) rectangle (5.8,2.5); 
        
         \draw[->](0,1.25) -- (12,1.25);  
          
          \fill[black!15]  (5,0) rectangle (5.8,2.5); 
          \fill[black]  (3,1.25) circle (0.5mm); 
          \fill[black]  (8,1.25) circle (0.5mm); 
          
           \draw[->](3,1.25) -- (3.5,1.75);  
            \draw[->](8,1.25) -- (7.7,1.85);  
            \draw[->, dashed](8,1.25) -- (8.5,1.75); 
         
         \draw[->, dotted](3,1.25) -- (3,3.25); 
         \draw[->, dotted](8,1.25) -- (8,3.25); 
         
                    \draw[->, very thick](5.4,1.) -- (5.4,1.6);
             \draw (5.4,1.55) node[above] { \small $\vec{M}$};
             
             \draw   (3,1.25)  node[below] { \small $e^{-}$};
                 \draw  (8,1.25) node[below] { \small $e^{-}$};
                 \draw  (8,1.75) node[above] { \small $P_{+}$};
                \draw  (3.45,1.75) node[above] { \small $P_{-}$};
                 
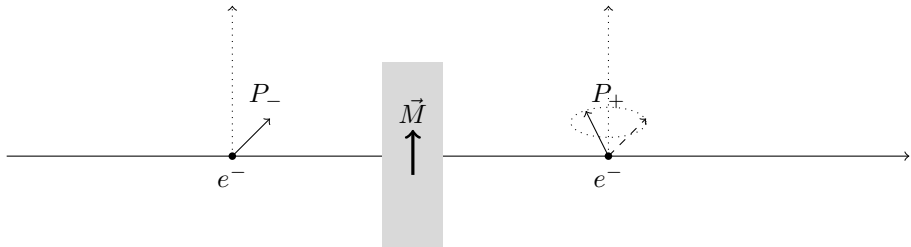
                \draw (5,-0.7) node[below] { \small Figure 2. The Siegmann experiment};
    \draw[dotted] (8,1.7) ellipse (0.5cm and 0.2cm);

        \end{tikzpicture}

\end{center}

 If the energy of incoming
electrons is neither too high nor to low, they can occupy the extended states of an empty band of the film to traverse the film, and the rate of absorption of electrons by the film during transmission is small; (i.e., the number of outgoing electrons is essentially the same as the number of electrons in the incoming beam). If the luminosity of the incoming beam is small, the electrons in the beam can be assumed to be independent. Hence it suffices to develop the
scattering theory of a single electron. The incoming electron is prepared in a pure state, i.e., one given 
by a normalized vector in $L^{2}(\mathbb{R}^3)\otimes \mathbb{C}^2$. But the state of an outgoing electron, after transmission through the film, is mixed and, hence, is described by a density matrix in
$\mathcal{J}^{+}_{1}(L^{2}(\mathbb{R}^3)\otimes \mathbb{C}^2)$. This is because the interaction of the electron with the degrees of freedom of the film lead to entanglement of the electron state with the state of the film. When the degrees of freedom of the film are ``traced out'' the state of the electron is, in general, mixed. During the time when the electron traverses the film its spin precesses around the direction of spontaneous magnetization $\vec{M}$ with a very large angular velocity. This precession is caused by a Zeeman-type interaction of the electron spin with the so-called ``Weiss exchange field'' that describes the ferromagnetic order inside the film. Furthermore, the direction of spin of the electrons tends to relax slowly towards the direction of spontaneous magnetization of the film, which is a consequence of interactions with spin waves in the film and of a small rate of absorption of electrons with spin opposite to the majority spin in the film. Thus, the reduced time evolution of the state of an electron is not unitary, but  can be approximated by a suitably chosen Lindblad dynamics. (For a theoretical description of these experiments see \cite{Albert-Fr}.)

\subsection{Scattering theory describing particle capture}\label{subsec:capture}

The mathematical concepts introduced so far do not suffice to describe systems of particles that can be captured (absorbed) by the target. But, as the example just described suggests, this possibility should be included in a general theory. Definitions of modified outgoing wave operators taking into account the possibility of capture have been proposed and can be found in the literature; see \cite{Alicki1,Davies2}. Here we follow essentially \cite{Davies2}. We suppose that the Lindblad operator has the form 
\begin{equation}\label{eq:def1_L}
  \mathcal{L} = \mathrm{ad}(H_0) + \mathrm{ad}(V) - \frac{i}{2}  \{C^{*}C,\,\cdot \,\} + i C \,\cdot \,C^{*}.
\end{equation}
The operators $H_0$ and $V$ act on a Hilbert space $\mathcal{H}$ and are self-adjoint; $H_0$ generates the unitary dynamics of a free particle, and $V$ describes \textit{static} interactions of the particle with the target. In contrast, the operator $C \in \mathcal{B}(\mathcal{H})$ is used to describe interactions of the particle with \textit{dynamical} degrees of freedom of the target. We suppose that $V$ and $C^* C$ are relatively compact with respect to $H_0$; so that, in particular,
\begin{equation*}
H_V := H_0 + V,
\end{equation*}
is self-adjoint on $\mathcal{H}$, with domain $\mathcal{D}( H_V ) = \mathcal{D}( H_0 )$.\\

 We require the following assumptions.
\begin{assumption}\label{V0}
The spectrum of $H_0$ is purely absolutely continuous, the singular continuous spectrum of $H_V$ is empty, and $H_V$ has at most finitely many eigenvalues of finite multiplicity. The wave operators 
\begin{equation*}
W_\pm( H_V , H_0 ) := \underset{t\to \mp \infty }{\slim} e^{ i t H_V } e^{ - i t H_0 } , \qquad W_{Ê\pm } ( H_0 , H_V ) := \underset{t\to \mp \infty }{\slim} e^{ i t H_0 } e^{ - i t H_V } \Pi_{ \mathrm{ac} }( H_V ) ,
\end{equation*}
exist on $\mathcal{H}$ and are asymptotically complete, in the sense that
\begin{align*}
& \mathrm{Ran} ( W_\pm( H_V , H_0 ) ) = \mathrm{Ran} ( \Pi_{ \mathrm{ac} }( H_V ) ) = \mathrm{Ran} ( \Pi_{ \mathrm{pp} }( H_V ) )^\perp , \\
& \mathrm{Ran} ( W_\pm( H_0 , H_V ) ) = \mathcal{H} .
\end{align*}
Here $\Pi_{ \mathrm{ac} }( H_V )$ and $\Pi_{ \mathrm{pp} }( H_V )$ denote the projections onto the absolutely continuous and pure point spectra of $H_V$, respectively.
\end{assumption}
\begin{assumption}\label{V1}
There exists a positive constant $\mathrm{c}_V$, depending on $C$ and $H_V$, such that
\begin{align}\label{eq:ZV}
\int_{ \mathbb{R} } \big \|ÊC e^{ - i t H_V } \Pi_{ \mathrm{ac} }( H_V ) u \big \|_{ \mathcal{H} }^2 dt \le \mathrm{c}_V^2 \| \Pi_{ \mathrm{ac} }( H_V ) u \|_{ \mathcal{H} }^2 ,
\end{align}
for all $u \in \mathcal{H}$.
\end{assumption}
In the example where $H_0 = - \Delta$ on $L^2( \mathbb{R}^3 )$ and $V$ is a potential, conditions on $V$ that imply Assumptions \ref{V0} and \ref{V1} are well-known; (see \cite{RS-III,JensenKato,BeMa92_01,Rauch}, and Section \ref{subsec:examples} for examples).

We are now prepared to introduce a modified outgoing wave operator allowing for the phenomenon of capture of the particle by the target; see \cite{Davies2}. As above, we consider the auxiliary (dissipative) operator
\begin{equation}\label{eq:def1_H}
H := H_V - \frac{i}{2} C^* C \equiv H_0 + V - \frac{i}{2} C^* C .
\end{equation}
We define the subspace $\mathcal{H}_{ \mathrm{b} }( H )$ as the closure of the vector space generated by the set of eigenvectors of $H$ corresponding to \emph{real} eigenvalues. It is not difficult to verify that
\begin{equation*}
\mathcal{H}_{ \mathrm{b} }( H ) = \mathcal{H}_{ \mathrm{pp} }( H_V ) \cap \mathrm{Ker}( C ) = \mathcal{H}_{ \mathrm{b} }( H^* ),
\end{equation*}
see \cite{Davies1}.
We also set
\begin{align*}
& \mathcal{H}_{ \mathrm{d} }( H ) := \big \{ u \in \mathcal{H} : \lim_{ t \to \infty } \|Êe^{ - i t H }Êu \|_{ \mathcal{H} } = 0 \big \},  \\
& \mathcal{H}_{ \mathrm{d} }( H^* ) := \big \{ u \in \mathcal{H} : \lim_{ t \to \infty } \|Êe^{ i t H^* }Êu \|_{ \mathcal{H} } = 0 \big \}.
\end{align*}

We define the modified wave operator $\tilde{\Omega}^-( \mathcal{L}_0 , \mathcal{L} )$ by
\begin{equation}\label{eq:modifwaveoperator}
\tilde{\Omega}^-( \mathcal{L}_0 , \mathcal{L} ) := \underset{t\to +\infty }{\slim} \: e^{it \mathcal{L}_0}  \big ( \Pi e^{-it\mathcal{L}}  ( \cdot ) \Pi \big ) ,
\end{equation}
where $\mathcal{L}_{0} := \text{ad}(H_{0})$, and where $\Pi$ is the orthogonal projection onto the orthogonal complement of
$\mathcal{H}_{ \mathrm{b} }( H ) \oplus \mathcal{H}_{ \mathrm{d} }( H )$.

\begin{theorem}\label{thm:capture}
Suppose that Assumptions \ref{V0} and \ref{V1} hold with $\mathrm{c}_V < 2$. Then the modified wave operator $\tilde{\Omega}^-( \mathcal{L}_0 , \mathcal{L} )$ exists on $\mathcal{J}_1( \mathcal{H} )$. For all 
$\rho \in \mathcal{J}_1^+( \mathcal{H} )$ with $\mathrm{tr}( \rho ) = 1$, we have that
\begin{equation*}
0 \le \mathrm{tr}( \tilde{\Omega}^-( \mathcal{L}_0 , \mathcal{L} ) \rho ) \le 1 ,
\end{equation*}
and $\mathrm{tr}( \tilde{\Omega}^-( \mathcal{L}_0 , \mathcal{L} ) \rho )$ is interpreted as the probability that the particle initially in the state $\rho$ eventually escapes from the target.
\end{theorem}
A key ingredient of the proof of Theorem \ref{thm:capture} is the following result on the scattering theory for dissipative operators, which is of some interest in its own right.
\begin{theorem}\label{thm:disswavemodified0}
Suppose that Assumptions \ref{V0} and \ref{V1} hold with $\mathrm{c}_V < 2$. Then the wave operator
\begin{equation*}
W_+( H , H_0 ):= \underset{t\to + \infty }{\slim} e^{ - i t H } e^{ i t H_0 },
\end{equation*}
exists on $\mathcal{H}$, is injective and its range is equal to
\begin{equation}\label{eq:range0}
\mathrm{Ran}( W_+( H , H_0 ) ) = \big ( \mathcal{H}_{ \mathrm{b} }( H ) \oplus \mathcal{H}_{ \mathrm{d} }( H^* ) \big )^\perp.
\end{equation}
\end{theorem}
We observe that, under the assumptions of Theorem \ref{thm:disswavemodified0}, 
$\mathrm{Ran}( W_+( H , H_0 ) )$ is closed, which is the main property used in the proof of Theorem \ref{thm:capture}. The inclusion $\mathrm{Ran}( W_+( H , H_0 ) ) \subset ( \mathcal{H}_{ \mathrm{b} }( H ) \oplus \mathcal{H}_{ \mathrm{d} }( H^* ) )^\perp$ is easily verified. It would be interesting to find conditions implying that the converse inclusion holds, too, without assuming a bound such as $\mathrm{c}_V < 2$. 

We also mention that, for Schr{\"o}dinger operators, a particular case of \eqref{eq:range0} has been recently proven by Wang and Zhu \cite{WangZhu} under the assumption that the imaginary part, $C^*C$, of $H$ is a short range potential whose norm is smaller than $\varepsilon$, for some $\varepsilon >0$.

\subsection{Comparison with the literature and organization of the paper}

Scattering theory for quantum dynamical semigroups has been studied previously in \cite{Davies2,Alicki1,AlickiFrigerio,reuteler}. The general ideas of the approach developed in this paper have been pioneered by Davies \cite{Davies4,Davies1,Davies2}. However, the abstract model we study and the kind of assumptions underlying our analysis significantly differ from those in \cite{Davies4,Davies1,Davies2}. The model considered in \cite{Davies2}  involves a Lindblad generator of the form $\mathcal{Z} = \mathcal{Z}_0 + \mathcal{Z}_1 + \mathcal{Z}_2$ acting on the space, 
$\mathcal{J}_1 ( L^2( \mathbb{R}^3 ) \otimes \mathcal{H}_1 )$, of trace-class operators on the Hilbert 
space $L^2( \mathbb{R}^3 ) \otimes \mathcal{H}_1$, where $\mathcal{H}_1$ is some Hilbert space, 
$\mathcal{Z}_0 = \mathrm{ad}( - \Delta \otimes {\bf1})$ generates 
the dynamics of a free particle, $\mathcal{Z}_1= {\bf1} \otimes Z_1$, where $Z_1$ 
is a Lindblad operator of the form \eqref{Lili0} acting on
$\mathcal{J}_1 ( \mathcal{H}_1 )$, and $\mathcal{Z}_2$ is an interaction term. Suitable assumptions are made on $Z_1$ and $\mathcal{Z}_2$, and the proofs rely on Cook's method and the Kato-Birman theory.

In this paper we consider a more general class of Lindblad operators. Moreover, we heavily rely on the Kato smoothness estimates stated in Assumptions \ref{Z0} and \ref{V1}. We think that assumptions of the kind introduced in this paper are well-suited to study the scattering theory for Lindblad operators. Besides, in many concrete situations, one is able to verify Assumptions \ref{Z0} and \ref{V1} using standard tools of spectral theory. As far as we know, our results on the completeness and invertibility of the wave operators stated in Theorem \ref{thm:main0} do not appear to have been previously described in the literature.

Our paper is organized as follows. Sections \ref{sec:dissipative} and \ref{sec:Lindblad} are devoted to the proof of Theorem \ref{thm:main0}. In Section \ref{sec:dissipative}, we study scattering theory for the dissipative operator $H$, which is the main ingredient of the analysis presented in Section \ref{sec:Lindblad}, namely the study of scattering theory for Lindblad operators. In Section \ref{sec:example}, we describe a concrete model that can be analyzed with the help of Theorem \ref{thm:main0}. 
In Section \ref{sec:capture}, we study the phenomenon of capture and prove Theorem \ref{thm:capture}. 
To render our paper reasonably self-contained, we review some technical details, including various known results, in appendices. \\

\noindent{\bf Acknowledgments.} The research of J.F. is supported in part by ANR grant ANR-12-JS01-0008-01. The research of B.S. is supported in part by ``Region Lorraine''.

\section{Scattering theory for dissipative perturbations of self-adjoint operators}\label{sec:dissipative}
In our approach to the scattering theory of Lindblad operators, an important role is played by the auxiliary dissipative operator
\begin{equation}
H := H_0 - \frac{i}{2} C^* C
\end{equation}
acting on a Hilbert space $\mathcal{H}$, as already mentioned in the last section. Our main concern in this section is to study the wave operators
\begin{align}\label{eq:defwaves}
W_{ \pm }( H , H_0 ) := \underset{t\to \mp \infty }{\slim} e^{ i t H } e^{ - i t H_0 } , \qquad  W_{ \pm }( H_0 , H ) := \underset{t\to \mp \infty }{\slim} e^{ i t H_0 } e^{ - i t H }
\end{align}
and to elucidate some of their properties. For previous results concerning scattering theory for dissipative operators on Hilbert spaces we refer to \cite{Martin,Mochizuki,Davies4,Davies1,Simon2,Kadowaki}.

In this section we set $\| \cdot  \| = \| \cdot \|_{ \mathcal{H} }$ to simplify the notations.

\subsection{Basic facts about wave operators for $H$ and $H_0$}
We recall that $H_0$ is supposed to be a self-adjoint operator on $\mathcal{H}$. Its domain is denoted by $\mathcal{D}( H_0 )$. Since $C$ is assumed to be bounded, it follows that $H$ is closed with domain $\mathcal{D}( H ) = \mathcal{D}( H_0 )$. Moreover, $H$ is the generator of a one-parameter group, 
$\{ e^{ - i t H } \}_{ t \in \mathbb{R} }$, of operators satisfying the a priori bound 
\begin{equation*}
\big \| e^{ - i t H } u \big \| \le e^{ \frac12 \| C^* C \| | t | } \| u \|, \quad t \in \mathbb{R} ,
\end{equation*}
(see e.g. \cite{Phillips1}). The subspaces $\mathcal{D}_\pm( H , H_0 )$ and $\mathcal{D}_\pm( H_0 , H )$ are defined as the sets of vectors in $\mathcal{H}$ such that the limits defining $W_\pm( H, H_0 )$ and $W_\pm( H_0 , H )$ exist. We recall the following basic facts about wave operators.
\begin{proposition}\label{prop:01}
Suppose that $W_\pm( H , H_0 )$ and $W_\pm( H_0 , H )$ exist on $\mathcal{D}_\pm( H , H_0 )$ and\\
$\mathcal{D}_\pm( H_0 , H )$, respectively. Then 
\begin{equation*}
e^{ - i t H } \mathcal{D}_\pm( H_0 , H ) \subset \mathcal{D}_\pm( H_0 , H ),  \quad 
e^{ - i t H_0} \mathcal{D}_\pm( H , H_0 ) \subset \mathcal{D}_\pm( H , H_0 ), 
\end{equation*}
for all $t \in \mathbb{R}$, and
  \begin{align}
    e^{ - i t H_0 }&W_\pm( H_0 , H ) = W_\pm( H_0 , H ) e^{ - i t H }\; \text{ on } \mathcal{D}_\pm( H_0 , H ), \label{eq:intertwin1} \\
    e^{ - i t H }& W_\pm( H , H_0 ) = W_\pm( H , H_0 ) e^{ - i t H_0 } \; \text{ on } \mathcal{D}_\pm( H , H_0 ). \label{eq:intertwin2}
  \end{align}
  Furthermore, 
  \begin{equation*}
  W_\pm ( H_0 , H )[ \mathcal{D}_\pm( H_0 , H ) \cap \mathcal{D}( H_0 ) ] \subset \mathcal{D}( H_0 ), \text{   } 
  W_\pm(H,H_0)[ \mathcal{D}_\pm( H , H_0 )\cap \mathcal{D}(H_0)]\subset \mathcal{D}(H_0), 
  \end{equation*}
  and
  \begin{align}
&\forall u \in \mathcal{D}_\pm( H_0 , H ) \cap \mathcal{D}( H_0 )\;,\; H_0 W_\pm( H_0 , H) u = W_\pm ( H_0 , H ) H u \; ; \label{eq:intertwin3}\\
   &\forall u \in \mathcal{D}_\pm( H , H_0 ) \cap \mathcal{D}( H_0 )\;,\; H W_\pm( H , H_0 ) u = W_\pm( H , H_0 ) H_0 u \; .\label{eq:intertwin4}
  \end{align}
\end{proposition}
\begin{proof}
The proof follows from standard arguments; (see the proof of Proposition \ref{prop:1} below.)
\end{proof}
In fact, since $\mathrm{Im} \langle u , H u \rangle = - \frac12 \| C u \|^2 \le 0$, for all $u \in \mathcal{D}( H )$, $H$ is dissipative, and hence the semi-group $\{ e^{ - i t H }Ê\}_{ t \ge 0 }$ is contractive,
\begin{equation}
\big \| e^{ - i t H } u \big \| \le \| u \|, \quad t \ge 0 , \label{eq:contrac}
\end{equation}
see, e.g., \cite{Engel}. In dissipative quantum scattering theory one studies the two wave operators 
$W_+( H , H_0 )$ and $W_-( H_0 , H )$. The contractivity of $\{ e^{ - i t H } \}_{ t \ge 0 }$ and unitarity of $\{ e^{ - i t H_0 } \}_{ t \in \mathbb{R} }$ show that $W_+( H , H_0 )$ and $W_-( H_0 , H )$ are contractions whenever they exist. In applications, the group $\{ e^{ - i t H_0 } \}_{ t \in \mathbb{R}Ê}$ is often given explicitly, and one can usually prove the existence of $W_+( H , H_0 )$ with the help of Cook's argument:
\begin{equation}
\begin{split}
& e^{ - i t H } e^{ i t H_0 } u = u - \frac12 \int_0^t e^{Ê- i s H }ÊC^* C e^{ i s H_0 }Êu ds, \\
& e^{ i t H_0 } e^{ - i t H } u = u - \frac12 \int_0^t e^{Êi s H_0 }ÊC^* C e^{Ê- i s H }Êu ds , \label{eq:cook0}
\end{split}
\end{equation}
for all $u \in \mathcal{H}$. A precise statement is the following proposition.
\begin{proposition}\label{prop:existenceCook}
Suppose that Assumption \ref{Y0} holds. Then $W_+( H , H_0 )$ exists on $\mathcal{H}$ and is injective.
\end{proposition}
\begin{proof}
The existence of $W_+( H , H_0 )$ is an obvious consequence of \eqref{eq:cook0} and Assumption \ref{Y0}. The injectivity is proven in \cite{Martin} or \cite{Davies1}, see also Appendix \ref{app:dissipative}.
\end{proof}
Next, we show that, if $C$ is $H_0$-smooth in the sense of Assumption \ref{Z0}, then $W_{ + }( H , H_0 )$ and $W_{ - }( H_0 , H )$ exist. The proof uses \eqref{eq:Z'0} together with a well-known argument.
\begin{proposition}\label{prop:existence}
Suppose that Assumption \ref{Z0} holds. Then the wave operators $W_{ + }( H , H_0 )$ and $W_{ - }( H_0 , H )$ exist on $\mathcal{H}$. Moreover $W_+( H , H_0 )$ is injective and $\mathrm{Ran}( W_-( H_0 , H ) )$ is dense in $\mathcal{H}$.
\end{proposition}
\begin{proof}
We establish existence of $W_-( H_0 , H )$; (existence of $W_{+}(H,H_0)$ is proven similarly). We use Cook's argument, see \eqref{eq:cook0}, and write 
\begin{align*}
 \Big \| \int_{t_1}^{t_2} e^{Êi s H_0 }ÊC^* C e^{Ê- i s H }Êu ds \Big \| 
&\le \sup_{Êv \in \mathcal{H} , \| v \|Ê= 1 } \int_{t_1}^{t_2} \big | \big \langle C e^{Ê- i s H_0 }Êv , C e^{Ê- i s H }Êu \big \rangle \big | ds \\
&\le \sup_{Êv \in \mathcal{H} , \| v \|Ê= 1 } \Big ( \int_{t_1}^{t_2} \big \| C e^{Ê- i s H_0 }Êv \|^2 ds \Big )^{\frac12} \Big ( \int_{t_1}^{t_2} \| C e^{Ê- i s H }Êu \|^2 ds \Big )^{ \frac12 } ,
\end{align*}
for all $u \in \mathcal{H}$, and for $0 < t_1 < t_2 < \infty$.
Since the two integrals on the right side converge on $[ 0 , \infty )$, by Assumption \ref{Z0} and \eqref{eq:Z'0}, we conclude that $\int_0^{t_n} e^{Êi s H_0 }ÊC^* C e^{Ê- i s H }Êu \, ds$ is a Cauchy sequence, for any sequence of times $(t_n)$ with $t_n \to \infty$, and hence that $W_-( H_0 , H )$ exists on 
$\mathcal{H}$. 

Injectivity of the wave operator $W_+( H , H_0 )$ is proven in Proposition \ref{prop:injectivity} of Appendix \ref{app:dissipative}.

To prove that $\mathrm{Ran}( W_-( H_0 , H ) )$ is dense in $\mathcal{H}$, we consider the adjoint wave operator $W_-( H^* , H_0 ) = \lim_{ t \to \infty } e^{ i t H^* } e^{ - i t H_0 }$. As in \eqref{eq:Z'0}, we have that
\begin{align}\label{eq:adjoint0}
\int_0^t \big \|ÊC e^{ i s H^* } u \big \|^2 ds = - \int_0^t \partial_s \big \langle u , e^{ - i s H } e^{ i s H^* } u \rangle ds = \| u \|^2 - \big \|Êe^{ i t H^*}u \big \|^2 \le \| u \|^2 ,
\end{align}
for all $u \in \mathcal{H}$. In the same way as for $W_+( H , H_ 0 )$, one can then verify that 
$W_-( H^* , H_0 )$ exists and is injective on $\mathcal{H}$. Using now that
\begin{equation*}
\mathrm{Ran}( W_-( H_0 , H ) )^\perp = \mathrm{Ker}( W_-( H^* , H_0 ) ) ,
\end{equation*}
we conclude that $\mathrm{Ran}( W_-( H_0 , H ) )^\perp = \{ 0 \}$, as claimed.
\end{proof}

\subsection{Smooth perturbations}
We will see that if the constant $\tilde{\mathrm{c}}_0$ in \eqref{eq:Z''0} is strictly less than $1$, or if Assumption \ref{Z0} holds with $\mathrm{c}_0 < 2$, then the four wave operators defined in \eqref{eq:defwaves} exist on $\mathcal{H}$, although $W_-( H , H_0 )$ and $W_+( H_0 , H )$ are in general not contractive. 

The results of this section are related to results of Kato \cite{Kato1}, whose results are more general, in the sense that he does not assume that $H_0$ is self-adjoint; it suffices to assume that the spectrum of $H_0$ is contained in the real axis. However, the proof in \cite{Kato1} requires the stronger assumption that $C$ is ``$H_0$-supersmooth'', (a terminology introduced in \cite{KatoYajima}), which means that 
$\sup_{ z \in \mathbb{C} \setminus \mathbb{R} } \| C ( H_0 - z )^{-1} C^* \| < \infty$. Kato's approach is stationary. In this paper, we employ a time-dependent method. We draw the reader's attention to a paper by Lin \cite{Lin}, which also follows a time-dependent approach, using a Dyson series, and is formulated in the general context of semi-groups in reflexive Banach spaces; (see, e.g., Evans \cite{Evans} for a generalization to non-reflexive Banach spaces). The assumptions in \cite{Lin} are stronger, though, and our proofs are much simpler, because we can take advantage of the Hilbert space formalism. 

We begin with proving that, if $\tilde{\mathrm{c}}_0$ in \eqref{eq:Z''0} is strictly less than $1$, then the inverse semigroup $\{ e^{i t H} \}_{ t \ge 0 }$ is uniformly bounded and $C$ is $H$-smooth.
\begin{lemma}\label{lm:ABCD}
Suppose that inequality \eqref{eq:Z''0} holds, with $\tilde{\mathrm{c}}_0 < 1$. Then the group $\{ e^{ -i t H } \}_{ t \in \mathbb{R} }$ is uniformly bounded,
\begin{equation}\label{eq:adjoint1}
\big \| e^{ - i t H }\big\|_{ \mathcal{B}( \mathcal{H} ) } \le ( 1 - \tilde{\mathrm{c}}_0^2 )^{-\frac12} , \quad t \in \mathbb{R}.
\end{equation}
Moreover, we have that
\begin{equation}\label{eq:ABC0}
\int_0^\infty \big \| C e^{ i t H }Êu \big \|^2 dt \le \frac{ \tilde{ \mathrm{c} }_0^2 }{ 1 - \tilde{\mathrm{c}}_0^2 } \| u \|^2 ,
\end{equation}
for all $u \in \mathcal{H}$.
Conversely, if there exists $m > 1$ such that
\begin{equation}\label{eq:adjointAB1}
\big \| e^{ - i t H }\big\|_{ \mathcal{B}( \mathcal{H} ) } \le m , \quad t \in \mathbb{R} ,
\end{equation}
then \eqref{eq:Z''0} is satisfied with $\tilde{\mathrm{c}}_0 = ( 1 - m^{-2} )^{1/2} < 1$.
\end{lemma}
\begin{proof}
Using \eqref{eq:Z'0}, we see that \eqref{eq:Z''0} is equivalent to
\begin{align*}
\big \|Êe^{ - i t H }Êu \big \|^2 \ge ( 1 - \tilde{\mathrm{c}}_0^2 ) \|Êu \|^2 ,
\end{align*}
for all $t \ge 0$ and all $u \in \mathcal{H}$. Equivalently,
\begin{align*}
\big \|Êe^{ i t H }Êu \big \| \le ( 1 - \tilde{\mathrm{c}}_0^2 )^{-\frac12} \|Êu \| ,
\end{align*}
for all $t \ge 0$ and all $u \in \mathcal{H}$. Therefore the assumption that \eqref{eq:Z''0} holds, for some
$\tilde{\mathrm{c}}_0 < 1$, is equivalent to the assumption that \eqref{eq:adjoint1} is satisfied, for all $t \in \mathbb{R}$. The statement that \eqref{eq:adjointAB1} implies \eqref{eq:Z''0} with $\tilde{\mathrm{c}}_0 = ( 1 - m^{-2} )^{1/2}$ is proven in the same way.

The bound \eqref{eq:ABC0} follows by noticing that
\begin{equation*}
\int_0^t \big \| C e^{ i s H }Êu \big \|^2 ds = \int_0^t \partial_s \big \| e^{ i s H }Êu \big \|^2 ds = \big \| e^{ i t H }Êu \big \|^2 - \| u \|^2 .
\end{equation*}
\end{proof}
The previous lemma allows us to establish the invertibility of the wave operators, and therefore the similarity of $H$ and $H_0$.
\begin{theorem}\label{thm:invertibility}
Suppose that Assumption \ref{Z0} is satisfied and that \eqref{eq:Z''0} holds with $\tilde{\mathrm{c}}_0 < 1$. Then the wave operators $W_\pm( H , H_0 )$ and $W_\pm( H_0  ,H )$ exist on $\mathcal{H}$ and are invertible in $\mathcal{B}( \mathcal{H} )$ and are inverses of each other,
\begin{align}
W_{ \pm } ( H , H_0 )^{-1} = W_{ \pm }( H_0 , H ). \label{eq:invert}
\end{align}
Moreover, the four wave operators leave the domain $\mathcal{D}( H_0 ) = \mathcal{D}( H )$ invariant, and the following intertwining property holds on $\mathcal{D}(H)$:
\begin{align}\label{eq:intertwing}
H  = W_{ \pm }( H , H_0 ) H_0 W_{ \pm } ( H_0 , H ).
\end{align}
\end{theorem}
\begin{proof}
Existence of $W_{ + }( H , H_0 )$ and $W_{ - }( H_0 , H )$ follows from Proposition \ref{prop:existence}. Existence of $W_{ - }( H , H_0 )$ and $W_{ + }( H_0 , H )$ can be proven in the same way, using inequality \eqref{eq:ABC0} in Lemma \ref{lm:ABCD}, instead of \eqref{eq:Z'''0}. 

The uniform boundedness of the operators $\{ e^{ - i t H_0 } \}_{ t \in \mathbb{R} }$, $\{ e^{ - i t H } \}_{ t \in \mathbb{R} }$ proven in Lemma \ref{lm:ABCD} implies that $W_\pm( H_0 , H )$ and $W_\pm( H , H_0 )$ are bounded operators on $\mathcal{H}$.

The invertibility of the wave operators is an an easy consequence of their definitions and of the uniform boundedness of $\{ e^{ - i t H_0 } \}_{ t \in \mathbb{R} }$, $\{ e^{ - i t H } \}_{ t \in \mathbb{R} }$. As an example, we can write
\begin{align*}
u &= e^{ i t H }Êe^{ - i t H_0Ê} e^{Êi t H_0 } e^{ - i t H } u \\
&= e^{Êi t H } e^{ - i t H_0 } W_{\pm} ( H_0 , H ) u + o ( 1 ) \\
&=  W_{\pm}( H , H_0 ) W_{\pm} ( H_0 , H ) u + o ( 1 ) ,
\end{align*}
as $t \to \mp \infty$, for all $u \in \mathcal{H}$. This shows that $W_{\pm}( H , H_0 ) W_{\pm} ( H_0 , H ) = \mathrm{Id}$. In the same way we can prove that $W_{\pm}( H_0 , H ) W_{\pm} ( H , H_0 ) = \mathrm{Id}$, and hence \eqref{eq:invert} holds.

The intertwining property follows from Proposition \ref{prop:01}.
\end{proof}
To prove the next result we require Assumption \ref{Z0} to hold, with $\mathrm{c}_0 < 2$. A simple argument will show that in this case also, the conclusions of Theorem \ref{thm:invertibility} hold.
\begin{theorem}\label{thm:bijec}
Suppose that Assumption \ref{Z0} holds with $\mathrm{c}_0 < 2$. Then, for all $u \in \mathcal{H}$,
\begin{equation}
\big \| e^{ - i t H } u \big \| \le \frac{ 1 }{ 1 - \mathrm{c}_0 / 2 } \| u \|, \quad t \in \mathbb{R} . \label{eq:unif-bound}
\end{equation}
In particular the conclusions of Theorem \ref{thm:invertibility} hold.
\end{theorem}
\begin{proof}
Let $w \in \mathcal{H}$. By \eqref{eq:cook0},
\begin{align*}
\big \|Êe^{ - i t H }Êw \big \|Ê&= \big \| e^{ i t H_0 } e^{ - i t H } w \big \| \\
&\ge \| w \| - \frac{1}{2} \Big\| \int_0^t e^{Êi s H_0 }ÊC^* C e^{Ê- i s H }Êw ds \Big\| \\ 
& \ge \| w \|Ê- \frac{1}{2} \sup_{Êv \in \mathcal{H} , \| v \|Ê= 1 } \Big ( \int_0^\infty \big \| C e^{Ê- i s H_0 }Êv \|^2 ds \Big )^{\frac12} \Big ( \int_0^\infty \| C e^{Ê- i s H }Êw \|^2 ds \Big )^{ \frac12 }  \\
&\ge \Big ( 1 -  \frac{1}{2} \mathrm{c}_0 \Big ) \| w \| ,
\end{align*}
for all $t \ge 0$, where we used Eqs. \eqref{eq:Z0} and \eqref{eq:Z'''0}. Applying this inequality to $w = e^{ i t H }Êu$ proves \eqref{eq:unif-bound} for $t \le 0$. For $t \ge 0$, \eqref{eq:unif-bound} is obvious by \eqref{eq:contrac}.

By Lemma \ref{lm:ABCD}, \eqref{eq:unif-bound} implies that \eqref{eq:Z''0} holds with $\tilde{\mathrm{c}}_0 < 1$ and therefore the conclusions of Theorem \ref{thm:invertibility} hold.
\end{proof}
\begin{remark}
$\quad$
\begin{enumerate}
\item The existence and invertibility of the adjoint wave operators
\begin{align}
W_{ \pm } ( H_0 , H^* ) :=  \underset{t\to \mp \infty }{\slim} e^{ i t H_0 } e^{ - i t H^* } = W_{ \pm }( H , H_0 )^* , \notag \\
W_{ \pm }( H^* , H_0 ) := \underset{t\to \mp \infty }{\slim} e^{ i t H^* } e^{ - i t H_0 } = W_\pm ( H_0 , H )^*, \label{eq:defwavesadjoint}
\end{align}
can be proven with the same arguments as above. Of course, these wave operators are not unitary in general.
\item If Assumptions \ref{Z0} holds, with $\mathrm{c}_0 < 2$, then one can show that the wave operators admit the integral representations
\begin{equation*}
\big \langle W_\pm ( H_0 , H ) u , v \rangle = \langle u , v \rangle \pm \frac12 \int_0^\infty \big \langle C e^{ \mp i t H_0 } u , C e^{ \mp i t H } v \big \rangle dt ,
\end{equation*}
and
\begin{equation*}
\big \langle W_\pm ( H , H_0 ) u , v \rangle = \langle u , v \rangle \pm \frac12 \int_0^\infty \big \langle C e^{ \mp i t H } u , C e^{ \mp i t H_0 } v \big \rangle dt ,
\end{equation*}
for all $u , v \in \mathcal{H}$. The integrals on the right side converge, as follows from the Cauchy-Schwarz inequality.
\item If we make the further assumption that $C$ is ``$H_0$-supersmooth'' \cite{KatoYajima}, i.e., that
\begin{equation*}
\sup_{ z \in \mathbb{C} \setminus \mathbb{R} } \| C ( H_0 - z )^{-1} C^* \| =: \mathrm{d}_0 < \infty, \text{   with a constant   }\mathrm{d}_0 < 2, 
\end{equation*}
then the following representations hold:
\begin{equation*}
\big \langle W_\pm ( H_0 , H ) u , v \rangle = \langle u , v \rangle \mp \frac12 \int_{ \mathbb{R} } \big \langle C ( H_0 - ( \lambda \pm i 0 ) )^{-1} u , C ( H^* - ( \lambda \pm i 0 ) )^{-1} v \big \rangle d\lambda ,
\end{equation*}
and
\begin{equation*}
\big \langle W_\pm ( H , H_0 ) u , v \rangle = \langle u , v \rangle \mp \frac12 \int_{ \mathbb{R} } \big \langle C ( H - ( \lambda \pm i 0 ) )^{-1} u , C ( H_0 - ( \lambda \pm i 0 ) )^{-1} v \big \rangle d\lambda ,
\end{equation*}
for all $u , v \in \mathcal{H}$; see \cite{Kato1}.
\end{enumerate}
\end{remark}

We conclude this section with a comment on the notion of completeness of the wave operators. In \cite{Martin}, Martin defines completeness of the wave operators in dissipative quantum scattering theory as follows: Suppose, to simplify matters, that $C^*C$ is a relatively compact perturbation of $H_0$ and that 
$H$ has only a finite number of eigenvalues of finite multiplicity. Let $P$ denote the projection onto the direct sum of all eigenspaces. Then the wave operators $W_+( H , H_0 )$, $W_-( H^* , H_0 )$ are said to be complete iff
\begin{equation*}
\mathrm{Ran}( W_+ ( H , H_0 ) ) = ( \mathrm{Id}Ê- P ) \mathcal{H} , \qquad \mathrm{Ran}( W_- ( H^* , H_0 ) ) = ( \mathrm{Id}Ê- P^* ) \mathcal{H}.
\end{equation*}
A scattering operator is then defined by
\begin{equation*}
S ( H , H_0 ) := W_-( H_0 , H ) W_+( H , H_0 ) \equiv \underset{t\to + \infty }{\slim} e^{Êi t H_0 } e^{ - 2 i t H } e^{Êi t H_0 }.
\end{equation*}
It follows from \cite{Davies1} that, under some further assumptions, an equivalent condition yielding the bijectivity of $S( H , H_0 )$ on $\mathcal{H}$ is that the subspace $\mathrm{Ran} ( W_+ ( H , H_0 ) )$ is closed. If Assumption \ref{Z0} holds, with $\mathrm{c}_0 < 2$, then, by Theorem \ref{thm:invertibility}, the wave operators $W_+( H , H_0 )$ and $W_-( H^* , H_0 )$ are complete and the scattering operator 
$S ( H , H_0 )$ is bijective on $\mathcal{H}$.

\section{Scattering theory for Lindblad operators}\label{sec:Lindblad}

Recall that the Lindblad operators studied in this paper have the form
\begin{equation*}
  \mathcal{L} = \mathrm{ad}(H_0) - \frac{i}{2}  \{C^{*}C,\,(\cdot) \,\} + i C \,(\cdot) \,C^{*} \equiv \mathcal{L}_0 - \frac{i}{2}  \{C^{*}C,\,(\cdot) \,\} + i C \,(\cdot) \,C^{*}.
\end{equation*}
To simplify our notation, we set
\begin{equation*}
\mathcal{W} := - \frac{i}{2}  \{C^{*}C,\,(\cdot) \,\} + i C \,(\cdot) \,C^{*}.
\end{equation*}
Recall that the trace norm in $\mathcal{J}_1( \mathcal{H} )$ is denoted by $\| \cdot \|_1$. The norm on the space $\mathcal{J}_2( \mathcal{H} )$ of Hilbert-Schmidt operators will be denoted by $\| \cdot \|_2$.

 \subsection{Existence and basic properties of $\Omega^+( \mathcal{L} , \mathcal{L}_0 )$}
 
We begin our considerations by stating a basic ``intertwining property'' of wave operators whose proof is standard, but,  for the convenience of the reader, is sketched below. We recall that $ \Omega^+(\mathcal{L},\mathcal{L}_0)$ and $ \Omega ^-(\mathcal{L}_0,\mathcal{L})$ are defined in \eqref{eq:defOmega_1}--\eqref{eq:defOmega_2}.
\begin{proposition}\label{prop:1}
Suppose that $ \Omega^+(\mathcal{L},\mathcal{L}_0)$ and $ \Omega ^-(\mathcal{L}_0,\mathcal{L})$ exist on $\mathcal{J}_1( \mathcal{H} )$ and $\mathcal{D}$, respectively, where $\mathcal{D}$ has been defined in \eqref{eq:dp}. Then  
  \begin{align}
&    e^{-it\mathcal{L}}  \Omega^+(\mathcal{L},\mathcal{L}_0)= \Omega^+(\mathcal{L},\mathcal{L}_0)e^{-it\mathcal{L}_0}\; \text{ on } \mathcal{J}_1( \mathcal{H} ), \label{eq:intert1} \\
 &   e^{-it\mathcal{L}_0} \Omega ^-(\mathcal{L}_0,\mathcal{L})= \Omega ^-(\mathcal{L}_0,\mathcal{L}) e^{-it\mathcal{L}} \; \text{ on } \mathcal{D}. \label{eq:intert2}
  \end{align}
  Furthermore, $$ \Omega^+(\mathcal{L},\mathcal{L}_0)[ \mathcal{D}(\mathcal{L}_0) ]\subset \mathcal{D}(\mathcal{L}_0), \quad  \Omega ^-(\mathcal{L}_0,\mathcal{L})[\mathcal{D} \cap \mathcal{D}(\mathcal{L}_0)]\subset \mathcal{D}(\mathcal{L}_0),$$
and
  \begin{align}
     &\forall \rho^+\in \mathcal{D}(\mathcal{L}_0)\;,\;\mathcal{L}\Omega^{+}(\mathcal{L},\mathcal{L}_0)\rho^+=\Omega^{+}(\mathcal{L},\mathcal{L}_0)\mathcal{L}_0 \rho^+\; ;\label{eq:intert3} \\
   &\forall \rho^-\in \mathcal{D} \cap \mathcal{D}(\mathcal{L}_0)\;,\;\mathcal{L}_0\Omega ^-(\mathcal{L}_0,\mathcal{L})\rho^-= \Omega ^-(\mathcal{L}_0,\mathcal{L}) \mathcal{L} \rho^- .\label{eq:intert4} 
  \end{align}
\end{proposition}
\begin{proof}
We only verify statements \eqref{eq:intert2} and \eqref{eq:intert4}. For $\rho^- \in \mathcal{D}$ and an arbitrary fixed $t \ge 0$, we have that
\begin{equation*}
e^{ i s \mathcal{L}_0 } e^{-is \mathcal{L}} e^{-it\mathcal{L}} \rho^- = e^{ - i t \mathcal{L}_0 } e^{ i (t + s ) \mathcal{L}_0 } e^{-i (t + s ) \mathcal{L}} \rho^-.
\end{equation*}
Taking $s \to \infty$ implies \eqref{eq:intert2}. The proof of \eqref{eq:intert1} is identical.

Next, we prove \eqref{eq:intert4}. Since $\{ e^{ - i t \mathcal{L} } \}_{ t \ge 0 }$ and $\{ e^{ i t \mathcal{L}_0 } \}_{ t \in \mathbb{R} }$ leave $\mathcal{D}( \mathcal{L} ) = \mathcal{D}( \mathcal{L}_0 )$ invariant, we obviously have that $\Omega^{-}( \mathcal{L}_0,\mathcal{L})[ \mathcal{D} \cap \mathcal{D}(\mathcal{L}_0)]\subset \mathcal{D}(\mathcal{L}_0)$. We then obtain, applying \eqref{eq:intert2} to $\rho^- \in \mathcal{D} \cap \mathcal{D}(\mathcal{L}_0)$, that
\begin{equation*}
t^{-1} \big ( e^{-it\mathcal{L}_0} - \mathrm{Id} \big )  \Omega^-(\mathcal{L}_0,\mathcal{L}) \rho^- = \Omega^-(\mathcal{L}_0,\mathcal{L}) t^{-1} \big ( e^{-it\mathcal{L}} - \mathrm{Id} \big ) \rho^-.
\end{equation*}
Passing to the limit $t \to 0$ yields \eqref{eq:intert4}. The proof of \eqref{eq:intert3} is identical.
\end{proof}

The existence of the wave operator $\Omega^+( \mathcal{L} , \mathcal{L}_0 )$ is the content of the next theorem. Our proof of this result, using Assumption \ref{Y0}, is close to the one in \cite{Davies2}.
But our proof of existence of $\Omega^+( \mathcal{L} , \mathcal{L}_0 )$, using Assumption \ref{Z0} instead of Assumption \ref{Y0}, appears to be new.
\begin{theorem}\label{thm:existenceOmega-}
Suppose that either Assumption \ref{Y0} holds, or that Assumption \ref{Z0} holds, with $\mathrm{c}_0 < 2$. Then $\Omega^+( \mathcal{L} , \mathcal{L}_0 )$ exists on $\mathcal{J}_1( \mathcal{H} )$.
\end{theorem}
\begin{proof}
We first assume that Assumption \ref{Y0} holds. Since $\mathcal{E}$ is dense in $\mathcal{H}$, the set of (finite) linear combinations of projections $ | u_i \rangle \langle u_i |$, with $u_i \in \mathcal{E}$, is dense in $\mathcal{J}_1( \mathcal{H} )$. Let $\rho = \sum_{i=1}^n \lambda_i  | u_i \rangle \langle u_i |$ be such a linear combination. Clearly
\begin{align}
e^{ - i t \mathcal{L} } e^{ i t \mathcal{L}_0 } \rho &= \rho - i \int_0^t e^{ - i s \mathcal{L} } \mathcal{W} e^{ i s \mathcal{L}_0 } \rho \notag \\
& = \rho + \int_0^t e^{ - i s \mathcal{L} } \Big ( -\frac12 \big ( C^* C  e^{ i s H_0 } \rho e^{ - i s H_0 } + e^{ i s H_0 } \rho e^{ - i s H_0 } C^* C \big ) \notag \\
& \qquad \qquad \qquad \qquad + C e^{ i s H_0 } \rho e^{ - i s H_0 } C^* \Big ) ds. \label{eq:AA1}
\end{align}
We now show that the above integrals converge in the norm of $\mathcal{J}_1( \mathcal{H} )$, uniformly in $t$. Using that the semi-group $\{ e^{ - i s \mathcal{L} } \}_{ s \ge 0 }$ is uniformly bounded on $\mathcal{J}_1( \mathcal{H} )$ by $2$, we write
\begin{align}
\big \|Êe^{ - i s \mathcal{L} } C^* C e^{ i s H_0 } \rho e^{ - i s H_0 } \big \|_1& \le 2 \sum_{i=1}^n | \lambda_i | \big \| C^* C e^{ i s H_0 } | u_i \rangle \langle u_i | e^{ - i s H_0 } \big \|_1 \notag \\ 
&\le 2 \sum_{i=1}^n | \lambda_i | \big \| C^* C e^{ i s H_0 } u_i \big \|_{ \mathcal{H} }Ê\| u_i \|_{ \mathcal{H} } . \label{eq:EFGH}
\end{align}
The second inequality follows from the Cauchy-Schwarz inequality, using that for any orthonormal basis $(e_j)$ in $\mathcal{H}$,
\begin{align*}
\sum_{ j \in \mathbb{N} } \big | \langle e_j , C^* C e^{ i s H_0 } u_i \rangle \langle u_i , e^{ - i s H_0 } e_j \rangle \big | &\le \Big ( \sum_{ j \in \mathbb{N} } \big | \langle e_j , C^* C e^{ i s H_0 } u_i \rangle \big |^2 \Big )^{\frac12} \Big ( \sum_{ j \in \mathbb{N} } \big | \langle u_i , e^{ - i s H_0 } e_j \rangle \big |^2 \Big )^{\frac12} \\
& = \big \| C^* C e^{ i s H_0 } u_i \big \|_{ \mathcal{H} }Ê\| u_i \|_{ \mathcal{H} }.
\end{align*}
 Since $s \mapsto \big \| C^* C e^{ i s H_0 } u_i \big \|_{ \mathcal{H} }$ is integrable on $[ 0 , \infty )$, by Assumption \ref{Y0}, Eq. \eqref{eq:EFGH} implies that the function
$$s \mapsto \big \|Êe^{ - i s \mathcal{L} } C^* C e^{ i s H_0 } \rho e^{ - i s H_0 } \big \|_1$$
 is also integrable on $[ 0 , \infty )$. The same argument shows that $s \mapsto \big \|Êe^{ - i s \mathcal{L} } e^{ i s H_0 } \rho e^{ - i s H_0 } C^* C \big \|_1$
  is integrable on $[ 0 , \infty )$ as well.

To bound the third term in \eqref{eq:AA1}, we notice that
\begin{align*}
\big \|Êe^{ - i s \mathcal{L} } C e^{ i s H_0 } \rho e^{ - i s H_0 } C^* \big \|_1 & \le 2 \sum_{i = 1}^n |Ê\lambda_i | \big \| C e^{ i s H_0 } | u_i \rangle \langle u_i | e^{ - i s H_0 } C^* \big \|_1 \\
& = 2 \sum_{i = 1}^n |Ê\lambda_i | \mathrm{tr} \big ( C e^{ i s H_0 } | u_i \rangle \langle u_i | e^{ - i s H_0 } C^* \big ) \\
& \le 2 \sum_{i = 1}^n |Ê\lambda_i | \big \| C^* C e^{ i s H_0 } | u_i \rangle \langle u_i | e^{ - i s H_0 } \big \|_1,
\end{align*}
and we have used the cyclicity of the trace. Therefore $s \mapsto \big \|Êe^{ - i s \mathcal{L} } C e^{ i s H_0 } u^* u e^{ - i s H_0 } C^* \big \|_1$ is integrable on $[ 0 , \infty )$.

Combining the previous estimates, we have shown that
\begin{equation}
\int_0^\infty \big \| e^{ - i s \mathcal{L} } \mathcal{W} e^{ i s \mathcal{L}_0 } \rho \big \|_1 ds < \infty .
\end{equation}
The proof is concluded by appealing to a density argument.

Next, we suppose that Assumption \ref{Z0} holds, with $\mathrm{c}_0 < 2$. Using the linearity of $e^{ - i t \mathcal{L} } e^{ i t \mathcal{L}_0 }$ and the fact that any $\rho \in \mathcal{J}_1( \mathcal{H} )$ can be written as a linear combination of four positive operators, we see that it suffices to prove the existence of $\lim e^{ - i t \mathcal{L} } e^{ i t \mathcal{L}_0 } \rho$, as $t \to \infty$, for any $\rho \in \mathcal{J}_1^+ ( \mathcal{H} )$. Thus, we let $\rho \in \mathcal{J}_1^+ ( \mathcal{H} )$ and write $\rho = u^* u$, for some $u \in \mathcal{J}_2( \mathcal{H} )$.

Let
\begin{align*}
  \mathcal{L}_1 := \mathrm{ad}( H ) \equiv H( \cdot ) - (  \cdot )H^* ,
  \end{align*}
with domain $\mathcal{D}( \mathrm{ad}( H ) )= \mathcal{D}(\text{ad}(H_0)) \subset \mathcal{J}_1( \mathcal{H} )$. For $t \ge 0$ and $\rho \ge 0$, we write
\begin{align*}
e^{ - i t \mathcal{L}Ê} e^{ i t \mathcal{L}_0 } \rho = e^{ - i t \mathcal{L}Ê} e^{ i t \mathcal{L}_1 } e^{ - i t \mathcal{L}_1 } e^{ i t \mathcal{L}_0 } \rho.
\end{align*}
By Theorems \ref{thm:invertibility} and \ref{thm:bijec}, we know that $\{ e^{  i t H } \}$ is uniformly bounded, for $t \in \mathbb{R}$, and that $W_+ ( H , H_0 ) = \slim e^{ -i t H }Êe^{ i t H_0 }$ ($t\to\infty$) exists on $\mathcal{H}$. This implies that $\{ e^{  i t \mathcal{L}_1 } \}$ is uniformly bounded, for $t \in \mathbb{R}$, and that $\slim e^{ -i  t \mathcal{L}_1 }Êe^{ i t \mathcal{L}_0 }$ ($t\to\infty$) exists on $\mathcal{J}_1( \mathcal{H} )$. Indeed, since $\rho = u^* u$, $u \in \mathcal{J}_2( \mathcal{H} )$, we find using Theorem \ref{thm:bijec} that
\begin{equation}
\big \| e^{  i t \mathcal{L}_1Ê} \rho \big \|_1 = \big \| e^{  i t H } u^* \big \|^2_2 \le \big \|Êe^{  i t H } \big \|_{Ê\mathcal{B}( \mathcal{H} ) }^2 \|Êu^* \|_2^2 \le \Big ( \frac{2}{ 2 - \mathrm{c}_0 } \Big )^2 \|Ê\rho \|_1 , \quad t \in \mathbb{R}. \label{eq:BBC1}
\end{equation}
To see that $\slim e^{ -i  t \mathcal{L}_1 }Êe^{ i t \mathcal{L}_0 }$ exists on $\mathcal{J}_1( \mathcal{H} )$, we observe that
  \begin{align*}
   e^{ - i t \mathcal{L}_1 } ( e^{i t \mathcal{L}_0 } \rho) = e^{ - i t H } e^{ i t H_{0} } \rho e^{ - i t H_0 } e^{ i t H^* },
  \end{align*}
  and therefore
  \begin{align}
    \big \| e^{ - i t H } e^{i t H_0 } \rho e^{ - i t H_0 } e^{ i t H^* } - W_+ \rho W_+^* \big \|_1 \to 0 , \quad t \to \infty . \label{eq:BBB1}
  \end{align}
  To simplify our notations, we set $W_+ \equiv W_+( H , H_0 )$ in the previous equation and throughout the rest of the proof. Statement \eqref{eq:BBB1} follows from
  \begin{align*}
    & \big \| e^{ - i t H } e^{ i t H_{ 0 } } \rho e^{ - i t H_{ 0 } } e^{ i t H^* } - W_+ \rho W_+^* \big \|_1 \\
    & = \big \| e^{ - i t H } e^{i t H_{ 0 } } u^* u e^{ - i t H_{ 0 } } e^{ i t H^* } - W_+ u^* u W_+^* \big \|_1 \\
    & \le \big \| ( e^{ - i t H } e^{i t H_{ 0 } } - W_+ ) u^* u e^{ - i t H_{ 0 } } e^{ i t H^* } - W_+ u^* u ( W_+^* - e^{ - i t H_{ 0 } } e^{ i t H^* } ) \big \|_1 \\
    & \le \big \| ( e^{ - i t H } e^{i t H_{ 0 } } - W_+ ) u^* \big \|_2 \| u \|_2 + \| u^* \|_2 \big \| u ( W_+^* - e^{ - i t H_{ 0 } } e^{ i t H^* } ) \big \|_2.
  \end{align*}
The right side is seen to tend to $0$, as $t \to \infty$, by recalling the isomorphism
$\mathcal{J}_2( \mathcal{H} ) \simeq \mathcal{H} \otimes \mathcal{H}$.

Equations \eqref{eq:BBC1} and \eqref{eq:BBB1} imply that
\begin{equation}
e^{ - i t \mathcal{L}Ê} e^{ i t \mathcal{L}_1 } e^{ - i t \mathcal{L}_1 } e^{ i t \mathcal{L}_0 } \rho = e^{ - i t \mathcal{L}Ê} e^{ i t \mathcal{L}_1 } ( W_+ \rho W_+^* ) + o ( 1 ) , \quad t \to \infty. \label{eq:BBD1}
\end{equation}
Next, we prove that $e^{ - i t \mathcal{L}Ê} e^{ i t \mathcal{L}_1 }$ converges strongly on 
$\mathcal{J}_1( \mathcal{H} )$, as $t \rightarrow \infty$. For any $\rho = u^* u$, $u \in \mathcal{J}_2( \mathcal{H} )$, we have that
\begin{equation*}
e^{ - i t \mathcal{L}Ê} e^{ i t \mathcal{L}_1 } \rho = \rho + \int_0^t e^{ - i s \mathcal{L}Ê} C ( e^{ i s \mathcal{L}_1 } \rho ) C^* ds.
\end{equation*}
We then use that
\begin{equation}
\big \|Êe^{ - i s \mathcal{L}Ê} C ( e^{ i s \mathcal{L}_1 } \rho ) C^* \big \|_1 \le 2 \big \| C ( e^{ i s \mathcal{L}_1 } \rho ) C^* \big \|_1 = 2 \big \|ÊC e^{ i s H } u \big \|_2^2 , \label{eq:CC4}
\end{equation}
and Theorem \ref{thm:bijec} together with Lemma \ref{lm:ABCD} tells us that $s \mapsto \big \|ÊC e^{ i s H } u \big \|_2^2$ is integrable on $[ 0 , \infty )$. (This follows again from the isomorphism 
$\mathcal{J}_2( \mathcal{H} ) \simeq \mathcal{H} \otimes \mathcal{H}$.) Therefore
\begin{equation*}
\Omega^+( \mathcal{L} , \mathcal{L}_1 ) := \underset{ t \to \infty }{ \slim } e^{ - i t \mathcal{L}Ê} e^{ i t \mathcal{L}_1 } 
\end{equation*}
exists on $\mathcal{J}_1^+( \mathcal{H} )$, hence on $\mathcal{J}_1( \mathcal{H} )$. We then deduce from \eqref{eq:BBD1} that $\Omega^+( \mathcal{L} , \mathcal{L}_0 )$ exists on $\mathcal{J}_1^+( \mathcal{H} )$ and satisfies
\begin{equation*}
\Omega^+( \mathcal{L} , \mathcal{L}_0 ) = \Omega^+( \mathcal{L} , \mathcal{L}_1 ) \Omega^+( \mathcal{L}_1 , \mathcal{L}_0 ) = \Omega^+( \mathcal{L} , \mathcal{L}_1 ) ( W_+ ( H , H_0 )\text{  } (\cdot) \text{  } W_+^*( H , H_0 ) ).
\end{equation*}
\end{proof}
\begin{remark}
Using Lemma \ref{lm:ABCD}, the above proof shows that, in the statement of Theorem \ref{thm:existenceOmega-}, the hypothesis that Assumption \ref{Z0} holds, with $\mathrm{c}_0 < 2$, can be replaced by the weaker hypothesis that Assumption \ref{Z0} holds, with $\tilde{\mathrm{c}}_0 < 1$, where 
$\tilde{\mathrm{c}}_0$ is defined in \eqref{eq:Z''0}.
\end{remark}

\subsection{Existence of $\Omega^-( \mathcal{L}_0 , \mathcal{L} )$}

We prove the existence of $\Omega^-( \mathcal{L}_0 , \mathcal{L} )$ following arguments in \cite[Theorem 4]{Davies2}, with some modifications. 
\begin{lemma}\label{thm:existenceOmega+}
Suppose that the map $s \mapsto \big \|ÊC (e^{ - i s \mathcal{L} } \rho) C^* \big \|_1$ is integrable on $[0,\infty)$, for all $\rho$ in a dense subset of $\mathcal{J}_1^+( \mathcal{H} )$. Then $\Omega^-( \mathcal{L}_0 , \mathcal{L} )$ exists on $\mathcal{J}_1( \mathcal{H} )$.
\end{lemma}
\begin{proof}
 As above, we set
  \begin{align*}
  \mathcal{L}_1 = \mathrm{ad}( H ) \equiv H( \cdot) - (\cdot )H^* ,
  \end{align*}
with domain $\mathcal{D}(\text{ad}(H)) = \mathcal{D}( \mathrm{ad}( H_{0} ) ) \subset \mathcal{J}_1( \mathcal{H} )$. We write
  \begin{align}\label{eq:A1}
    e^{i t \mathcal{L}_0 } e^{ - i t \mathcal{L} } =  e^{i t \mathcal{L}_0 } e^{ - i t \mathcal{L}_1 } + e^{i t \mathcal{L}_0 } \big ( e^{ - i t \mathcal{L} } - e^{ - i t \mathcal{L}_1 } \big ).
  \end{align}
As in the proof of Theorem \ref{thm:existenceOmega-}, it suffices to prove strong convergence of $e^{ i t \mathcal{L}_0 } e^{ - i t \mathcal{L} }$ on the cone of positive operators. Thus, let $\rho \in \mathcal{J}_1^+( \mathcal{H} )$ belong to a dense subset as in the statement of the lemma and decompose $\rho = u^* u$, with $u \in \mathcal{J}_2( \mathcal{H} )$. By the same arguments as in \eqref{eq:BBB1}, we have that
  \begin{align}
    \big \| e^{i t H_0 } e^{ - i t H } \rho e^{ i t H^* } e^{ - i t H_0 } - W_- \rho W_-^* \big \|_1 \to 0 , \quad \text{as    }t \to \infty , \label{eq:BB1}
  \end{align}
with $W_- \equiv W_-( H_0 , H )$.
 
 Next, we treat the second term in \eqref{eq:A1}. We write
  \begin{align*}
    e^{i t \mathcal{L}_0 } \big ( e^{ - i t \mathcal{L} } - e^{ - i t \mathcal{L}_1 } \big ) \rho & =  e^{i t \mathcal{L}_0 } \int_0^t  e^{ - i (t-s) \mathcal{L}_1 } C (e^{ - i s \mathcal{L} } \rho ) C^* ds \\
    & = \int_0^t e^{i s \mathcal{L}_0 } e^{i (t-s) \mathcal{L}_0 } e^{ - i (t-s) \mathcal{L}_1 } C (e^{ - i s \mathcal{L} } \rho) C^* ds .
  \end{align*}
  For any fixed $s \geq 0$, we have that
  \begin{align*}
    \lim_{ t \to \infty } e^{i s \mathcal{L}_0 } e^{i (t-s) \mathcal{L}_0 } e^{ - i (t-s) \mathcal{L}_1 } C (e^{ - i s \mathcal{L} } \rho ) C^* = e^{i s \mathcal{L}_0 } W_- C (e^{ - i s \mathcal{L} } \rho ) C^* W_-^*
  \end{align*}
  in $\mathcal{J}_1( \mathcal{H} )$.
 
 The existence of the limit
  \begin{align*}
    \lim_{ t \to \infty } \int_0^t e^{i s \mathcal{L}_0 } e^{i (t-s) \mathcal{L}_0 } e^{ - i (t-s) \mathcal{L}_1 } C (e^{ - i s \mathcal{L} } \rho) C^* ds = \int_0^\infty e^{i s \mathcal{L}_0 } W_- C (e^{ - i s \mathcal{L} } \rho) C^* W_-^* ds ,
  \end{align*}
then follows from the dominated convergence theorem, since
  \begin{align*}
    \mathds{1}_{[0,t]}(s) \big \| e^{i s \mathcal{L}_0 } e^{i (t-s) \mathcal{L}_0 } e^{ - i (t-s) \mathcal{L}_1 } C (e^{ - i s \mathcal{L} } \rho) C^* \big \|_1 \le \mathds{1}_{ [0,\infty )}( s) \big \|ÊC (e^{ - i s \mathcal{L} } \rho) C^* \big \|_1, 
  \end{align*}
and since the map $s \mapsto \big \|ÊC (e^{ - i s \mathcal{L} } \rho) C^* \big \|_1$ is  integrable on $[0,\infty)$, by assumption. 

  Summarizing, we have shown that, for all
  $\rho$ in a dense subset of $\mathcal{J}_1^+( \mathcal{H} )$, 
  \begin{align*}
    \lim_{ t \to \infty } e^{i t \mathcal{L}_0 } e^{ - i t \mathcal{L} } \rho = W_- \rho W_-^* + \int_0^\infty e^{i s \mathcal{L}_0 } W_- C (e^{ - i s \mathcal{L} } \rho) C^* W_-^* ds 
  \end{align*}
  in $\mathcal{J}_1( \mathcal{H} )$. By a density argument, the existence of the limit $\lim_{ t \to \infty } e^{i t \mathcal{L}_0 } e^{ - i t \mathcal{L} } \rho$ extend to all $\rho \in \mathcal{J}_1^+( \mathcal{H} )$,  and this concludes the
  proof.
\end{proof}
\begin{remark}
The terms 
$$W_- ( H_0 , H ) \rho W_-( H_0 , H )^* \text{   and    }\text{   }\Omega^-( \mathcal{L}_0 , \mathcal{L} ) \rho - W_- ( H_0 , H ) \rho W_-( H_0 , H )^* $$
 of the decomposition
 \begin{align*}
\Omega^-( \mathcal{L}_0 , \mathcal{L} ) \rho = W_-( H_0 , H ) \rho W_-( H_0 , H )^* + \int_0^\infty e^{i s \mathcal{L}_0 } W_-( H_0 , H ) C (e^{ - i s \mathcal{L} } \rho) C^* W_-( H_0 , H )^* ds ,
  \end{align*}
appearing the in the proof of the previous lemma are usually referred to as the elastically and inelastically scattered components of $\rho$. 
\end{remark}
\begin{theorem}
Suppose that the wave operator $W_-( H_0 , H )$ defined in \eqref{eq:defwaves} exists on $\mathcal{H}$, is injective and has closed range. Then $\Omega^-( \mathcal{L}_0 , \mathcal{L} )$ exists on $\mathcal{J}_1( \mathcal{H} )$. In particular, if Assumption \ref{Z0} holds, with $\mathrm{c}_0 < 2$, (or, more generally, if Assumption \ref{Z0} holds and $\tilde{\mathrm{c}}_0<1$, where $\tilde{\mathrm{c}}_0$ is defined in \eqref{eq:Z''0}) then $\Omega^-( \mathcal{L}_0 , \mathcal{L} )$ exists on $\mathcal{J}_1( \mathcal{H} )$.
\end{theorem}
\begin{proof}
As before, it suffices to prove strong convergence of $e^{ i t \mathcal{L}_0 } e^{ - i t \mathcal{L} }$ on the cone of positive operators. By Lemma \ref{thm:existenceOmega+}, it suffices to show that the map $s \mapsto \big \|ÊC (e^{ - i s \mathcal{L} } \rho) C^* \big \|_1$ is integrable on $[0,\infty)$. We use again the notation $W_- \equiv W_-( H_0 , H )$. Since, by assumption, $ W_- $ is injective, with closed range, there exists a positive constant $c$ such that
  $\| W_- \varphi \| \ge c \| \varphi \|$, for all  $\varphi \in \mathcal{H}$.  Consequently, for all $\rho \in \mathcal{J}_1( \mathcal{H} )$, $\rho \ge 0$,
  \begin{align*}
    \big \|ÊC(e^{ - i s \mathcal{L} } \rho)  C^* \big \|_1 \le c^{-2}Ê \big \|ÊW_- C (e^{ - i s \mathcal{L} } \rho)C^* W_-^* \big \|_1 = c^{-2}Ê \big \|Êe^{is \mathcal{L}_0 } W_- C (e^{ - i s \mathcal{L} } \rho ) C^* W_-^* \big \|_1.
  \end{align*}
  To prove this inequality, we use that
  \begin{align*}
    \big \|ÊC ( e^{ - i s \mathcal{L} } \rho)  C^* \big \|_1 = \big \| C ( e^{ - i s \mathcal{L} } \rho)^{\frac12} \big \|^2_2 ,
  \end{align*}
together with the isomorphism $\mathcal{J}_2( \mathcal{H} ) \simeq \mathcal{H} \otimes \mathcal{H}$.
  Using the intertwining relation $H_0 W_- = W_- H$, see Proposition \ref{prop:01}, we  observe that
  \begin{align*}
    e^{is \mathcal{L}_0 } W_- C (e^{ - i s \mathcal{L} } \rho ) C^* W_-^* = \partial_s e^{is \mathcal{L}_0 } W_- (e^{ - i s \mathcal{L} } \rho ) W_-^*.
  \end{align*}
  Therefore $s \mapsto \| e^{is \mathcal{L}_0 } W_- C (e^{ - i s \mathcal{L} } \rho ) C^* W_-^* \|_1$ is integrable on $[0,\infty)$; for
  \begin{align*}
    \int_0^t \big \|  e^{is \mathcal{L}_0 } W_- C (e^{ - i s \mathcal{L} } \rho ) C^* W_-^* \big \|_1 ds &=  \int_0^t \mathrm{tr} \big ( e^{is \mathcal{L}_0 } W_- C (e^{ - i s \mathcal{L} } \rho ) C^* W_-^* \big ) ds \\
    &=  \Big [ \mathrm{tr} \big ( e^{is \mathcal{L}_0 } W_- (e^{ - i s \mathcal{L} } \rho ) W_-^* \big ) \Big ]_0^t \\
    &= \mathrm{tr} \big ( W_- (e^{ - i t \mathcal{L} } \rho ) W_-^* \big ) - \mathrm{tr} \big ( W_- \rho W_-^* \big ) ,
  \end{align*}
  is uniformly bounded in $t \in [ 0 , \infty )$.
  
By Theorem \ref{thm:invertibility}, $W_-( H_0 , H )$ is a bijection on $\mathcal{H}$ if Assumptions \ref{Z0} is satisfied and \eqref{eq:Z''0} holds with $\mathrm{c}_0 < 1$. 
\end{proof}

\subsection{Asymptotic completeness of wave operators} \label{Asymc}
In this section we prove (asymptotic) completeness of the wave operators. We use again the notation
  \begin{align*}
  \mathcal{L}_1 = \mathrm{ad}( H ) \equiv H (\cdot) - (\cdot) H^* ,
  \end{align*}
The following Dyson-Phillips series \cite{Phillips1} converges in 
$\mathcal{B}( \mathcal{J}_1( \mathcal{H} ) )$, for all $t \in \mathbb{R}$:
\begin{align}
e^{ - i t \mathcal{L} } = e^{ - i t \mathcal{L}_1 } + \sum_{ n \ge 1 } \mathcal{S}_n ( t ) , \label{eq:DysonL1}
\end{align}
where, for all $n \in \mathbb{N}$,
\begin{align}
\mathcal{S}_n ( t ) \rho := \int_0^t \int_0^{s_1} \cdots \int_0^{s_{n-1}} & e^{ -i ( t - s_1) H }  C e^{ - i ( s_1 - s_2 ) H } C \cdots e^{ - i ( s_{ n -1 } - s_n ) H } C e^{ - i s_n H } \rho e^{ i s_n H^*} C^* \notag \\
& e^{ i ( s_{ n -1 } - s_n ) H^* } \cdots C^* e^{ i ( s_1 - s_2 )H^* } C^* e^{ i ( t - s_1 ) H^* }  ds_n \dots d s_1 , \label{eq:DysonL2}
\end{align}
for $\rho \in \mathcal{J}_1( \mathcal{H} )$. For all $t \in \mathbb{R}$ and all $n \in \mathbb{N}$, $\mathcal{S}_n(t) \in \mathcal{B}( \mathcal{J}_1( \mathcal{H} ) )$, and the series $\sum_{ n \ge 1 } \mathcal{S}_n ( t )$ converges normally in $\mathcal{B}( \mathcal{J}_1( \mathcal{H} ) )$.

\begin{lemma}\label{thm:unif-boundL}
Suppose that Assumption \ref{Z0} is satisfied and that \eqref{eq:Z''0} holds, with $\tilde{\mathrm{c}}_0 < 1/ \sqrt{2}$. Then, there exists a positive constant $\mathrm{d}_0$ such that, for all $\rho \in \mathcal{J}_1 ( \mathcal{H} )$,
\begin{equation}
\big \| e^{ - i t \mathcal{L} } \rho \big \|_1 \le \mathrm{d}_0 \| \rho \|_1, \quad t \in \mathbb{R} . \label{eq:unif-bound2}
\end{equation}
Moreover, there exists a constant $\tilde{\mathrm{d}}_0 > 0$ such that
\begin{equation}
\int_{ \mathbb{R} } \big \| C ( e^{ - i t \mathcal{L} } \rho ) C^* \big \|_1 dt \le \tilde{\mathrm{d}}_0 \| \rho \|_1. \label{eq:unif-bound3}
\end{equation}
\end{lemma}
\begin{proof}
It suffices to prove the lemma for $\rho$ in the cone of positive operators. Let $\rho \in \mathcal{J}_1^+( \mathcal{H} )$, $\rho = u^* u$, with $u \in \mathcal{J}_2( \mathcal{H} )$. For $t \ge 0$, we can choose $\mathrm{d}_0 = 2$ in \eqref{eq:unif-bound2} as mentioned in Remark \ref{rk:intro1}. We prove \eqref{eq:unif-bound2} for $t \le 0$. We estimate the terms in the Dyson series \eqref{eq:DysonL1}--\eqref{eq:DysonL2} as follows: By Lemma \ref{lm:ABCD}, we know that $\| e^{ - i t H } \|_{ \mathcal{B}( \mathcal{H} ) } \le 1 / ( 1 - \tilde{\mathrm{c}}_0^2 )^{1/2}$, which shows that
\begin{equation}
\big \| e^{ - i t \mathcal{L}_1 } \rho \big \|_1 \le  \frac{ 1 }{ 1 - \tilde{\mathrm{c}}_0^2 }  \| \rho \|_1. \label{eq:B0}
\end{equation}
The terms \eqref{eq:DysonL2} are then bounded by
\begin{align*}
\big \|Ê\mathcal{S}_n ( t ) \rho \big \|_1 \le \frac{ 1 }{ 1 - \tilde{\mathrm{c}}_0^2 } \int_0^t \int_0^{s_1} \cdots \int_0^{s_{n-1}} \big \| & C e^{ - i ( s_1 - s_2 ) H } \cdots C e^{ - i ( s_{ n -1 } - s_n ) H } \notag \\
& \quad C e^{ - i s_n H } u^* \big \|_2^2 \notag  ds_n \dots ds_1 . 
\end{align*}
Applying again Lemma \ref{lm:ABCD}, one obtains that
\begin{align}\label{eq:A0}
\int_{ - \infty }^0 \big \|ÊC e^{ - i t H } u \big \|_{ \mathcal{H} }^2 dt \le \frac{ \tilde{\mathrm{c}}_0^2 }{ 1 - \tilde{\mathrm{c}}_0^2 } \| u \|_{ \mathcal{H} }^2 ,
\end{align}
for all $u \in \mathcal{H}$. This in fact implies that
\begin{align}\label{eq:A10}
\int_{ \mathbb{R}Ê} \big \|ÊC e^{ - i t H } u \big \|_{ \mathcal{H} }^2 dt \le \frac{ \tilde{\mathrm{c}}_0^2 }{ 1 - \tilde{\mathrm{c}}_0^2 } \| u \|_{ \mathcal{H} }^2 ,
\end{align}
because, for $s \ge 0$,
\begin{align*}
\int_{ - \infty }^s \big \|ÊC e^{ - i t H } u \big \|_{ \mathcal{H} }^2 dt = \int_{ - \infty }^0 \big \|ÊC e^{ - i ( t + s ) H } u \big \|_{ \mathcal{H} }^2 dt  \le \frac{ \tilde{\mathrm{c}}_0^2 }{ 1 - \tilde{\mathrm{c}}_0^2 } \| e^{ - i s H }Êu \|_{ \mathcal{H} }^2 \le \frac{ \tilde{\mathrm{c}}_0^2 }{ 1 - \tilde{\mathrm{c}}_0^2 } \|Êu \|_{ \mathcal{H} }^2 ,
\end{align*}
where we use the contractivity of $e^{ - i s H }$ in the last inequality. Passing to the limit $s \to \infty$ gives \eqref{eq:A10}.

Now, applying \eqref{eq:A10} $n$ times, and using once again that $\mathcal{J}_2( \mathcal{H} ) \simeq \mathcal{H} \otimes \mathcal{H}$, we find that
\begin{align*}
\big \|Ê\mathcal{S}_n ( t ) \rho \big \|_1 \le \frac{ 1 }{ 1 - \tilde{\mathrm{c}}_0^2 } \Big ( \frac{ \tilde{\mathrm{c}}_0^2 }{ 1 - \tilde{\mathrm{c}}_0^2 } \Big )^{n} \| u^* \|_2^2 .
\end{align*}
Plugging \eqref{eq:B0}  and this estimate into Eqs. \eqref{eq:DysonL1}--\eqref{eq:DysonL2} yields the bound
\begin{equation}
\big \| e^{ - i t \mathcal{L} } \rho \big \|_1 \le \frac{ 1 }{ 1 - 2 \tilde{\mathrm{c}}_0^2 }  \| \rho \|_1 ,
\end{equation}
which proves \eqref{eq:unif-bound2}.

The proof of \eqref{eq:unif-bound3} follows similarly and is left to the reader.
\end{proof}
\begin{theorem}\label{thm:invertibility2}
Suppose that Assumption \ref{Z0} holds, with $\mathrm{c}_0 < 2 - \sqrt{2}$. Then the wave operators 
$\Omega^\pm( \mathcal{L} , \mathcal{L}_0 )$ and $\Omega^\pm( \mathcal{L}_0  , \mathcal{L} )$ exist on 
$\mathcal{J}_1( \mathcal{H} )$, are invertible in $\mathcal{B}( \mathcal{J}_1( \mathcal{H} ) )$ and are inverses of each other,
\begin{align}
\Omega^{ \pm } ( \mathcal{L} , \mathcal{L}_0 )^{-1} = \Omega^{ \pm }( \mathcal{L}_0 , \mathcal{L} ). \label{eq:invert2}
\end{align}
Moreover, these four wave operators leave $\mathcal{D}( \mathcal{L}_0 ) = \mathcal{D}( \mathcal{L} )$ invariant, and the following intertwining property holds:
\begin{align}\label{eq:intertwing2}
\mathcal{L}  = \Omega^{ \pm }( \mathcal{L} , \mathcal{L}_0 ) \mathcal{L}_0 \Omega^{ \pm } ( \mathcal{L}_0 , \mathcal{L} ).
\end{align}
\end{theorem}
\begin{proof}
It follows from Theorem \ref{thm:bijec} and Lemma \ref{lm:ABCD} that if $\mathrm{c}_0 < 2 - \sqrt{2}$ we can choose $\tilde{\mathrm{c}}_0 < 1 / \sqrt{2}$ in \eqref{eq:Z''0}. In particular, the conclusions of Lemma \ref{thm:unif-boundL} hold.

By Theorem \ref{thm:invertibility}, we know that the wave operators $W_\pm( H , H_0)$ and $W_\pm( H_0 , H )$ exist on $\mathcal{H}$. As in Statement \eqref{eq:BBB1} appearing in the proof of Theorem \ref{thm:existenceOmega-}, this implies that $\Omega^\pm( \mathcal{L}_1 , \mathcal{L}_0 )$ and $\Omega^\pm( \mathcal{L}_0  , \mathcal{L}_1 )$ exist on $\mathcal{J}_1( \mathcal{H} )$.

Next, using that $s \mapsto \| C ( e^{ - i s \mathcal{L}_1 } \rho ) C^* \|_1$ is integrable on $\mathbb{R}$, for all $\rho \in \mathcal{J}_1( \mathcal{H} )$, see Lemma \ref{lm:ABCD} and \eqref{eq:CC4}, and that $s \mapsto \| C ( e^{ - i s \mathcal{L} } \rho ) C^* \|_1$ is integrable on $\mathbb{R}$, by Lemma \ref{thm:unif-boundL}, we prove by using the same arguments as in the proof of Theorem \ref{thm:existenceOmega-} that the wave operators $\Omega^\pm( \mathcal{L} , \mathcal{L}_1 )$ and $\Omega^\pm( \mathcal{L}_1  , \mathcal{L} )$  exist on $\mathcal{J}_1( \mathcal{H} )$.

Since the groups $\{ e^{ - i t \mathcal{L}_0 } \}_{ t \in \mathbb{R} }$, $\{ e^{ - i t \mathcal{L}_1 } \}_{ t \in \mathbb{R} }$ and $\{ e^{ - i t \mathcal{L}} \}_{ t \in \mathbb{R} }$ are all uniformly bounded, it is then easy to prove that the wave operators $\Omega^\pm( \mathcal{L} , \mathcal{L}_0 )$ and 
$\Omega^\pm( \mathcal{L}_0  , \mathcal{L} )$ exist, using the ``chain rules''
\begin{align*}
& \Omega^\pm( \mathcal{L} , \mathcal{L}_0 ) = \Omega^\pm( \mathcal{L} , \mathcal{L}_1 ) \Omega^\pm( \mathcal{L}_1 , \mathcal{L}_0 ),  \qquad  \Omega^\pm( \mathcal{L}_0 , \mathcal{L} ) = \Omega^\pm( \mathcal{L}_0 , \mathcal{L}_1 ) \Omega^\pm( \mathcal{L}_1 , \mathcal{L} ) .
\end{align*}
Invertibility of the wave operators and \eqref{eq:invert2} are proven in the same way. The intertwining property follows as in the proof of Proposition \ref{prop:1}.
\end{proof}

\section{A concrete example}\label{sec:example}
\subsection{Choice of a model}
In this section, we study a concrete model of a particle scattering off a 
dynamical target, whose effective dynamics is
given by a master equation of Lindblad type. Pure states of the particle
are unit rays in the Hilbert space $L^2(\mathbb{R}^3) \otimes \mathfrak{h}$, where $\mathfrak{h}$ is a
complex separable Hilbert space used to describe internal degrees
of freedom of the particle, and mixed states are given by density matrices, (i.e., by operators of trace
$1$ in the convex cone of positive trace-class operators). The effective dynamics of the particle is approximated by a one-parameter semi-group generated by a Lindblad operator of the
form
\begin{equation}\label{Lindbla}
  \mathcal{L}:= \mathrm{ad}(-\Delta + H_{\mathrm{int}}) - \frac{i }{2}  \sum_{j \in J} \{ C^{*}_j C_j, (\cdot) \}  +  i  \sum_{j \in J}C_j (\cdot) C^{*}_j,  
\end{equation}
where $\mathrm{ad}(A)\rho:=A \rho - \rho A^*$, and $H_{\mathrm{int}}$ is a self-adjoint operator on 
$\mathfrak{h}$ describing the dynamics of the internal degrees of freedom of the particle. 
To simplify matters, we suppose that $\mathrm{dim}( \mathfrak{h} ) < \infty$, and, without loss of generality, we assume that $H_{ \mathrm{int} } \ge 0$. The Lindbladian $\mathcal{L}$
acts on the Banach space
$\mathcal{J}_1(L^2(\mathbb{R}^3) \otimes \mathfrak{h})$ of trace-class
operators on $L^2(\mathbb{R}^3) \otimes \mathfrak{h}$. Its domain is
denoted by $\mathcal{D}(\mathcal{L})$. In the following, we give conditions on the operators $C_j, j \in J,$ that guarantee the
existence of wave operators, and we prove asymptotic completeness
for certain choices of the $C_j$'s.

  We begin by explaining how to derive meaningful expressions for the
operators $C_j$, $j \in J$.  In many situations, the interaction of the
particle $P$ with the target causes decoherence over the spectrum
of an observable $A=A^{*}$ acting on the Hilbert space $\mathcal{H} = L^2(\mathbb{R}^3) \otimes \mathfrak{h}$ 
of the particle. In our model, we use that every density matrix $\rho$ on
$L^2(\mathbb{R}^3) \otimes \mathfrak{h} $ can be represented as a
kernel operator, 
\begin{equation} \label{argu} \rho:= \rho(x,x'),
\end{equation}
where $x,x' \in \mathbb{R}^3$, and
$$\rho(x,x') \in \mathfrak{h} \otimes \mathfrak{h}, \text{     because     }\text{  }
\mathcal{J}_1(L^2(\mathbb{R}^3) \otimes \mathfrak{h}) \subset
\mathcal{J}_2(L^2(\mathbb{R}^3) \otimes \mathfrak{h}) \simeq
L^2(\mathbb{R}^3\times \mathbb{R}^3 ; \mathfrak{h} \otimes \mathfrak{h}) ;
$$
(see Appendix \ref{kernel} for more details). The variable $x$ stands for the position of the particle. This representation is useful if the interaction of the particle with the target causes decoherence in particle position space. Alternatively, we may consider a model exhibiting decoherence over the spectrum of the momentum operator of the particle, replacing
$x$ and $x'$ in \eqref{argu} by the particle momentum variables $p$ and
$p'$. In the former case (i.e., if decoherence in position space arises), then
\begin{equation}
  (e^{-it \mathcal{L}} \rho)(x,x') \rightarrow \rho(x,x) \delta_{x, x'}
\end{equation}
as $t$ tends to $+ \infty$, as long as $x$ and $x'$ belong to the support of the target.  
A typical choice of a Lindblad generator, $\mathcal{L}_{\text{dec}}$, leading to
this asymptotic behavior is
\begin{equation} \label{ldec} 
\mathcal{L}_{\mathrm{dec}}\rho:=- i \lambda \sum_{j=1}^{3}[G_j, [G_j, \rho ]],
\end{equation}
where $[ A, B]= AB-BA$, $\lambda$ is a complex constant with
$\mathrm{Re}(\lambda)>0$, and $G_j$ is the operator of multiplication by $x_{j} g_{j}(x)$, where $x_{j}$ is
$j$-th component of the particle position, $x$, in standard Cartesian coordinates of $\mathbb{R}^{3}$, and $g_{j}(x)$, $j=1,2,3$ are functions identically equal to $1$ on the support of the target and decreasing rapidly to $0$, outside the target.
We note that
\begin{align*}
  [G_j, \rho](x,x')=(x_j g_j(x) -x'_j g_j(x')) \rho(x,x'),
\end{align*}
hence
\begin{equation}
  (\mathcal{L}_{\mathrm{dec}} \rho)(x,x')=- i \lambda \sum_{j=1}^{3} (x_j g_j(x) - x'_j g_j(x'))^2 \rho(x,x').
\end{equation}
We observe that
$ \mathcal{L}_{\mathrm{dec}}$ can be recast in the form of
\eqref{Lindbla}, because
\begin{equation}
  - i [G_j, [G_j, \rho]]= -i [G_j, G_j \rho- \rho G_j] = -i (G_j^2 \rho + \rho G_j^2) + 2 i G_j \rho G_j, 
\end{equation}
hence $C_j=G_j$, $j \in J \equiv \lbrace 1,2,3 \rbrace$.

If the time evolution of the density matrix $\rho$ were given by
\begin{equation}
  \partial_t \rho_t(x,x') = -i  (\mathcal{L}_{\mathrm{dec}} \rho_t)(x,x') \equiv - \lambda \vert x-x' \vert ^2 \rho_t(x,x'),
\end{equation}
whenever $x$ and $x'$ belong to the support of the target, we would deduce that the matrix elements 
$\rho_t(x,x')$, $x\neq x'$, with $x$ and $x'$ in the support of the target, of the density matrix $\rho$ decay exponentially fast in $t$, with a rate proportional to the square of the distance between $x$ and $x'$. 

 Of course decoherence can also arise in the internal space of the particle, i.e., for the internal degrees of freedom of $P$, in momentum space, or in momentum space \textit{and} position space, or in momentum space and/or position space and/or internal space. As in previous sections, we assume that $J=\{1\}$, 
 since this does not affect the nature of our conclusions, 
and we denote $C_1$ in \eqref{Lindbla} by $C$, throughout the rest of this section.
We consider three classes of examples:
\vspace{2mm}

\begin{itemize}
\item $C= g(X) \cdot X$, where $X$ is multiplication by $x=(x_1, x_2, x_3) \in \mathbb{R}^{3}$,
  $g: \mathbb{R}^3 \rightarrow \mathbb{C}^3$ is a function of
  rapid decay at infinity. This is a slightly simplified version of the example discussed above,
  where the interaction of the particle with the target is
  localized in space near the support of the target. It leads to (partial) decoherence 
  in position space. The Lindblad operator in \eqref{Lindbla} is then given by
  \begin{equation}
    \label{Lili}
    \mathcal{L}\rho= \mathrm{ad}(-\Delta + H_{\mathrm{int}} )\rho - i [g(X) \cdot X, [g(X) \cdot X, \rho]].
  \end{equation}
  \vspace{0mm}

\item $C$ depends non-trivially on internal degrees of freedom of the
  particle, e.g., on a component of the spin of the particle. If $\mathrm{dim}(\mathfrak{h})< \infty$
  a physically reasonable choice is
$$C=g(X)\cdot S,$$
where $g: \mathbb{R}^3 \rightarrow \mathbb{C}^3$ is a function that
vanishes rapidly at infinity, and $S$ is the spin
operator. The Lindblad operator in \eqref{Lindbla} is then given by
\begin{equation}\label{Lili2}
  \mathcal{L}\rho= \mathrm{ad}(-\Delta + H_{\mathrm{int}} + \beta B(X) \cdot S)\rho - i [g(X)\cdot S, [g(X)\cdot S, \rho]], 
\end{equation}
where $B(x) \in \mathbb{R}^{3}$ is the magnetic field at the point $x \in \mathbb{R}^{3}$, and 
$\beta$ is a coupling constant. The operator $ \beta B(X) \cdot S$ describes the
Zeeman term.  \vspace{4mm}

\item The interaction between the particle and the target may lead to
  decoherence in position space \textit{and} in momentum space. In this case, we may choose $C$ to be given by
$$C=g(X)\cdot (\alpha X + \beta P) f(P) + \mathrm{h.c.}$$
where $g: \mathbb{R}^3 \rightarrow \mathbb{C}^3$ and
$f: \mathbb{R}^3 \rightarrow \mathbb{C}$ are functions decreasing
rapidly at infinity.
\end{itemize}

\subsection{Validating abstract assumptions by imposing simple conditions on $C$}

In order to verify the assumptions of Theorem \ref{thm:main0} for our concrete choices of operators $C$, we appeal to a variety of known results. In what follows we discuss some examples.

Let $\langle X \rangle$ be the operator of multiplication by  $\sqrt{1+x^{2}}$. It is well-known that the map $t \mapsto \big \|Ê\langle X \rangle^{-1-\varepsilon } e^{ i t \Delta } \varphi \big \|$ is integrable on $\mathbb{R}$, for all 
$\varepsilon > 0$ and all $\varphi \in \mathcal{D}( \langle X \rangle^{1+\varepsilon} ) \subset L^2( \mathbb{R}^3 )$. This yields the following result.
\begin{proposition}
Suppose that $\| C^* C \langle X \rangle^{1+\varepsilon} \| < \infty$, for some $\varepsilon > 0$. Then $\Omega^+( \mathcal{L} , \mathcal{L}_0 )$ exists on $\mathcal{J}_1( L^2( \mathbb{R}^3 \otimes \mathfrak{h} ) )$.
\end{proposition}

The optimal Kato smoothness estimate
\begin{align*}
  & \int_{ \mathbb{R} } \big \| | X |^{-1} e^{ i t \Delta } \varphi \big \|^2 dt \le \pi \| \varphi \|^2 ,
\end{align*}
for all $\varphi \in L^2( \mathbb{R}^3 )$, is established in \cite{Simon}. Applying Theorem \ref{thm:main0}, we immediately arrive at the following proposition.
\begin{proposition}
Suppose that $\| C | X | \| < 2 \pi^{-1/2} $. Then $\Omega^+( \mathcal{L} , \mathcal{L}_0 )$ and $\Omega^-( \mathcal{L}_0 , \mathcal{L} )$ exist on $\mathcal{J}_1( L^2( \mathbb{R}^3 \otimes \mathfrak{h} ) )$.

If  $\|ÊC | X | \| <  ( 2 - \sqrt{2} ) \pi^{-1/2}$ then $\Omega^+( \mathcal{L} , \mathcal{L}_0 )$ and $\Omega^-( \mathcal{L}_0 , \mathcal{L} )$ exist and are complete.
\end{proposition}
In a similar way we may rely on the estimate \cite{Simon}:
\begin{align*}
  & \int_{ \mathbb{R} } \big \| \langle X \rangle^{-1} ( 1 - \Delta )^{\frac14} e^{ i t \Delta } \varphi \big \|^2 dt \le \frac{Ê\pi }{ 2 } \| \varphi \|^2 .
\end{align*}

Another possibility is to relate the operator $C$ to a potential from a large class, in particular to a Rollnik potential, using the estimate 
\begin{align}
  & \int_{ \mathbb{R} } \big \| D( X ) e^{ i t \Delta } \varphi \big \|^2 dt \le \frac{ \|D^2\|_{ \mathrm{R} } }{ 2 \pi } \| \varphi \|^2 , \label{eq:Rollnik}
\end{align}
for all $\varphi \in L^2( \mathbb{R}^3 )$, where $D( X )$ denotes the operator of multiplication by the real-valued Rollnik potential $D(x)$. We recall \cite{ReedSimon2} that a measurable function $D : \mathbb{R}^3 \to \mathbb{C}$ is called a Rollnik potential iff
\begin{equation*}
  \| D \|_{ \mathrm{R} }^2 := \int_{ \mathbb{R}^3 } \frac{ | D ( x ) | | D ( y ) | }{ | x - y |^2 } dx dy < \infty.
\end{equation*}
Estimate \eqref{eq:Rollnik} follows from the fact that for any real-valued Rollnik potential $D$, and for all $\kappa \in \mathbb{C}$ with $\mathrm{Re}( \kappa ) > 0$, the operator $D(X) ( - \Delta + \kappa^2 )^{-1} D(X)$,  has the kernel
  \begin{equation*}
    \frac{ D ( x )Êe^{ - \kappa | x - y | } D( y ) }{ 4 \pi | x - y | } ,
  \end{equation*}
and hence, for all $z \in \mathbb{C} \setminus \mathbb{R}$,
  \begin{equation*}
    \big \| D( X ) ( - \Delta - z )^{-1} D( X ) \big \| \le \frac{1}{ 4 \pi } \| D^2 \|_{ \mathrm{R} }.
  \end{equation*}
By \cite{Kato1}, this implies \eqref{eq:Rollnik}, and applying Theorem \ref{thm:main0}, we obtain the following result.
\begin{proposition}
Suppose that $D$ is a real-valued, invertible Rollnik potential such that $\| C D( X )^{-1}  \| \| D^2 \|_{ \mathrm{R} }^{1/2} < 8 \pi^{1/2}$. Then $\Omega^+( \mathcal{L} , \mathcal{L}_0 )$ and $\Omega^-( \mathcal{L}_0 , \mathcal{L} )$ exist on $\mathcal{J}_1( L^2( \mathbb{R}^3 \otimes \mathfrak{h} ) )$.

If  $\|ÊC D( X )^{-1} \| \| D^2 \|_{ \mathrm{R} }^{1/2} <  4 ( 2 - \sqrt{2} ) \pi^{1/2}$ then $\Omega^+( \mathcal{L} , \mathcal{L}_0 )$ and $\Omega^-( \mathcal{L}_0 , \mathcal{L} )$ exist and are complete.
\end{proposition}

Using the Hardy-Littlewood-Sobolev inequality, (see e.g. \cite{LiebLoss}), the previous proposition can be applied to the concrete examples of the previous subsection. Considering for instance the Lindblad operator of \eqref{Lili2}, we have:
\begin{corollary}
  Let   $g_j\in L^{3/2} (\mathbb{R}^3 )\cap L^{\infty} (\mathbb{R}^3 )$, $g_j>0$ almost everywhere for $j = 1 , 2 , 3$. If
  $ \lVert (\sum_jg_j)^{1/2} \rVert_{\infty}^{} \lVert \sum_jg_j
  \rVert_{3/2}^{1/2} < \bigl(3\lVert S
  \rVert_{}^{} \bigr)^{-1}\pi^{1/3}(\tfrac{2^{19}}{3})^{\frac{1}{6}}$,
  then $\Omega^+(\mathcal{L},\mathcal{L}_0)$ and
  $\Omega^-(\mathcal{L}_0,\mathcal{L})$ exist for the Lindblad-type
  operator of \eqref{Lili2}.

  If in addition,
  $\lVert (\sum_jg_j)^{1/2} \rVert_{\infty}^{} \lVert \sum_jg_j
  \rVert_{3/2}^{1/2} < \bigl(3\lVert S
  \rVert_{}^{} \bigr)^{-1}\pi^{1/3}(\tfrac{2^{19}}{3^{13}})^{\frac{1}{6}}$,
  then the wave operators are asymptotically complete.
\end{corollary}

\section{Scattering theory and particle capture}\label{sec:capture}

In this section we explain how the analysis of Sections \ref{sec:dissipative} and \ref{sec:Lindblad} can be modified to prove Theorem \ref{thm:capture}. As in Section \ref{sec:dissipative}, we set $\| \cdot  \| = \| \cdot \|_{ \mathcal{H} }$ to simplify the notations.

\subsection{Proof of Theorem \ref{thm:capture}}

We begin our proof by studying the wave operators for the dissipative operator $H$. We recall that the absolutely continuous subspace, $\mathcal{H}_{ \mathrm{ac} }( H )$, for the dissipative operator $H \equiv H_0  + V - i C^* C / 2$ can be defined as follows (\cite{Davies4,Davies1}): Let
\begin{equation*}
M( H ) := \Big \{ u \in \mathcal{H} , \exists \mathrm{c}_u > 0 ,  \forall v \in \mathcal{H} , \int_0^\infty \big |Ê\langle e^{ - i t H }Êu , v \rangle \big |^2 dt \le \mathrm{c}_u \| v \|^2 \Big \}.
\end{equation*}
Then $\mathcal{H}_{ \mathrm{ac} }( H ) := \overline{ M ( H ) }$ is the closure of $M( H )$ in $\mathcal{H}$. It is proven in \cite{Davies1} that
\begin{equation*}
\mathcal{H}_{ \mathrm{ac} }( H ) = \mathcal{H}_{ \mathrm{b} }( H )^\perp ,
\end{equation*}
where, we recall, $\mathcal{H}_{ \mathrm{b} }( H )$ denotes the closure of the set of eigenvectors of $H$ in $\mathcal{H}$. Moreover, if $u \in \mathcal{H}_{ \mathrm{ac} }( H )$ then
\begin{equation}\label{eq:limitAC}
\lim_{ t \to \infty } \langle e^{ - i t H }Êu , v \rangle = \lim_{ t \to \infty } \big \|ÊK e^{ - i t H }Êu \big \|Ê= 0,
\end{equation}
for all $v \in \mathcal{H}$ and all compact operators $K$ on $\mathcal{H}$; see \cite{Davies4}.

We also recall the definitions
\begin{align*}
& \mathcal{H}_{ \mathrm{d} }( H ) := \big \{ u \in \mathcal{H} , \lim_{ t \to \infty } \|Êe^{ - i t H }Êu \| = 0 \big \},  \quad \mathcal{H}_{ \mathrm{d} }( H^* ) := \big \{ u \in \mathcal{H} , \lim_{ t \to \infty } \|Êe^{ i t H^* }Êu \| = 0 \big \}.
\end{align*}
\begin{theorem}\label{thm:disswavemodified}
Suppose that Assumptions \ref{V0} and \ref{V1} hold, with $\mathrm{c}_V < 2$. Then the wave operator $W_+( H , H_0 ) = \underset{t\to + \infty }{\slim} e^{ - i t H } e^{ i t H_0 }$ exists on $\mathcal{H}$ and is injective, and its range is equal to
\begin{equation}\label{eq:range}
\mathrm{Ran}( W_+( H , H_0 ) ) = \big ( \mathcal{H}_{ \mathrm{b} }( H ) \oplus \mathcal{H}_{ \mathrm{d} }( H^* ) \big )^\perp.
\end{equation}
Moreover, the wave operator
\begin{equation*}
W_- ( H_0 , H ) := \underset{t\to + \infty }{\slim} e^{ i t H_0 } e^{ - i t H } \Pi_{ \mathrm{ac} } ( H ) 
\end{equation*}
exists on $\mathcal{H}$. (Here $\Pi_{ \mathrm{ac} } ( H )$ denotes the orthogonal projection onto the absolutely continuous subspace of $H$.)
\end{theorem}
\begin{proof}
We first prove that $W_+( H , H_0 )$ exists on $\mathcal{H}$. Let $u \in \mathcal{H}$. Using Assumption \ref{V0}, we write
\begin{equation*}
e^{ - i t H }Êe^{ i t H_0 } u = e^{ - i t H }Êe^{ i t H_V } e^{ - i t H_V }Êe^{ i t H_0 }Êu = e^{ - i t H } e^{ i t H_V } W_+( H_V , H_0 ) u + o(1) , \quad t \to \infty .
\end{equation*}
By Assumption \ref{V0} we also know that $W_+( H_V , H_0 )$ is a unitary operator from $\mathcal{H}$ to $\mathrm{Ran} ( \Pi_{ \mathrm{ac} }( H_V ) )$. Therefore it suffices to prove that 
\begin{equation*}
W_+( H , H_V ) := \underset{t\to + \infty }{\slim} e^{ - i t H } e^{ i t H_V } \Pi_{ \mathrm{ac} }( H_V ),
\end{equation*}
exists on $\mathcal{H}$ and is injective on $\mathrm{Ran} ( \Pi_{ \mathrm{ac} }( H_V ) )$, with closed range. Existence can be proven in the same way as in Proposition \ref{prop:existence}, using Cook's argument together with Assumption \ref{V1}. We then have that
\begin{align*}
&\big \|Êe^{ - i t H } e^{ i t H_V } \Pi_{ \mathrm{ac} }( H_V ) u \big \|Ê\\
&\ge \| \Pi_{ \mathrm{ac} }( H_V ) u \|Ê- \frac12 \Big \|Ê\int_0^t e^{ - i s H }ÊC^* C e^{ i s H_V } \Pi_{Ê\mathrm{ac}Ê}( H_V ) u ds \Big \| \\
&\ge \| \Pi_{ \mathrm{ac} }( H_V ) u \|Ê- \frac12 \sup_{ v \in \mathcal{H} , \| v \|Ê= 1 } \Big ( \int_0^t \big \|ÊC e^{ - i s H } v \big \|^2 ds \Big )^{\frac12} \Big (\int_0^t \big \| C e^{ i s H_V } \Pi_{Ê\mathrm{ac}Ê}( H_V ) u \big \|^2 ds \Big )^{\frac12} \\
& \ge ( 1 - \mathrm{c}_V / 2 ) \| \Pi_{ \mathrm{ac} }( H_V ) u \|,
\end{align*}
for all $u \in \mathcal{H}$.
Since $\mathrm{c}_V < 2$ by assumption, this shows that $W_+( H , H_V )$, and hence $W_+( H , H_0 )$, are injective, with closed ranges.

Next, we establish existence of $W_- ( H_0 , H )$. Since $\Pi_{Ê\mathrm{pp} }( H_V )$ is compact, we know that $\Pi_{ \mathrm{pp} }( H_V ) e^{ - i t H } \Pi_{ \mathrm{ac} } ( H ) \to 0$, as $t \to \infty$, by \eqref{eq:limitAC}. It therefore suffices to prove existence of
\begin{equation*}
\underset{t\to + \infty }{\slim} e^{ i t H_0 } \Pi_{ \mathrm{ac} }( H_V ) e^{ - i t H } \Pi_{ \mathrm{ac} } ( H )
\end{equation*}
on $\mathcal{H}$. Writing $e^{ i t H_0 } \Pi_{ \mathrm{ac} }( H_V ) e^{ - i t H } = e^{ i t H_0 } e^{ - i t H_V } \Pi_{ \mathrm{ac} }( H_V ) e^{ i t H_V } e^{ - i t H }$, one can proceed in the argument as above. This shows that $\slim e^{ i t H_0 } \Pi_{ \mathrm{ac} }( H_V ) e^{ - i t H }$, $t \to + \infty$, exists on $\mathcal{H}$, (and that its restriction to $\mathrm{Ran}( \Pi_{ \mathrm{ac} } ( H_V ) )$ is injective, with closed range). Therefore $W_-( H_0 , H )$ exists.

Finally we prove \eqref{eq:range}. From the definition of $\mathcal{H}_{ \mathrm{ac}Ê}( H )$ we see that $\mathrm{Ran}( W_+( H , H_0 ) ) \subset \mathcal{H}_{ \mathrm{ac} }( H )$. Indeed, if $ u = W_+ ( H , H_0 ) w \in \mathrm{Ran}( W_+( H , H_0 ) )$ the intertwining property implies that
\begin{equation*}
\int_0^\infty \big |Ê\langle e^{ - i t H }Êu , v \rangle \big |^2 dt = \int_0^\infty \big |Ê\langle e^{ - i t H_0 }Êw , W_+ (H , H_0)^* v \rangle \big |^2 dt \le \mathrm{const} \| W_+( H , H_0 ) \| \| w \|^2 \| v \|^2,
\end{equation*}
for all $v \in \mathcal{H}$, since $H_0$ has purely absolutely continuous spectrum. Hence $W_+( H , H_0 ) = \Pi_{ \mathrm{ac} }( H ) W_+( H , H_0 )$. In the same way as for $W_-( H_0 , H )$, one verifies that $W_+( H_0 , H^* )$ exists, and hence
\begin{equation}\label{eq:range2}
W_+( H , H_0 )^* = W_+( H_0 , H^* ).
\end{equation}
From the definitions of $W_+( H_0 , H^* )$ and $\mathcal{H}_{\mathrm{d}}( H^* )$ we obtain that
\begin{equation*}
\mathrm{Ker} ( W_+( H_0 , H^* ) ) = \mathcal{H}_{ \mathrm{ac} }( H )^\perp \oplus \big ( \mathcal{H}_{ \mathrm{ac} }( H ) \cap \mathcal{H}_{\mathrm{d}}( H^* ) \big ) .
\end{equation*}
Since $\mathcal{H}_{ \mathrm{ac} }( H )^\perp = \mathcal{H}_{ \mathrm{b} }( H )$, and since one can easily verify that $\mathcal{H}_{\mathrm{d}}( H^* ) \subset \mathcal{H}_{ \mathrm{b} }( H )^\perp$, this equation can be rewritten as
\begin{equation*}
\mathrm{Ker} ( W_+( H_0 , H^* ) ) = \mathcal{H}_{ \mathrm{b} }( H ) \oplus \mathcal{H}_{\mathrm{d}}( H^* ).
\end{equation*}
From \eqref{eq:range2} and the fact that $\mathrm{Ran}( W_+( H , H_0 ) )$ is closed we obtain \eqref{eq:range}.
\end{proof}

\begin{proof}[Proof of Theorem \ref{thm:capture}]
To prove Theorem \ref{thm:capture} with the help of Theorem \ref{thm:disswavemodified}, it suffices to follow and adapt \cite{Davies2} in a straightforward way. We do not present the details of the arguments.
\end{proof}

 \subsection{Example}\label{subsec:examples}
 
We consider Lindblad operators of the form introduced in Section \ref{sec:example}, but add a potential to the free dynamics of the particle. Thus we consider operators of the form
\begin{equation}\label{Lindbla2}
  \mathcal{L}:= \mathrm{ad}(-\Delta  + V(X)+ H_{\mathrm{int}} ) - \frac{i }{2} \{ C^{*} C, (\cdot) \}  +  i  C (\cdot) C^{*},  
\end{equation}
on $\mathcal{J}_1( L^2( \mathbb{R}^3 ) \otimes \mathfrak{h} )$, where $V(X)$ denotes the operator of multiplication by the real-valued function $V(x)$ on $L^2( \mathbb{R}^3 )$, $H_{\mathrm{int}}$ is a positive self-adjoint operator on $\mathfrak{h}$ and $C \in \mathcal{B}( L^2( \mathbb{R}^3 ) \otimes \mathfrak{h} )$. We give an example of conditions that imply our abstract Assumptions \ref{V0} and \ref{V1}. For instance, it suffices to suppose that, for some $\varepsilon> 0$ and for all $x \in \mathbb{R}$, $|ÊV( x )  |Ê\le \mathrm{const} \langle x \rangle^{-2-\varepsilon}$ to guarantee that Assumption \ref{V0} is satisfied. Of course, this condition is far from being optimal. If, in addition, $0$ is neither an eigenvalue nor a resonance of $H_V$ then it is known, (see \cite{BeMa92_01}), that, for any $\varepsilon > 0$, there exists a constant 
$\mathrm{c}_1 > 0$ such that
\begin{equation}\label{eq:EFG}
\int_{ \mathbb{R}Ê} \big \| \langle X \rangle^{-1-\varepsilon} e^{ - i t H_V } \Pi_{ \mathrm{ac} }( H_V ) u \|^2 dt \le \mathrm{c}_1^2Ê\| \Pi_{ \mathrm{ac} }( H_V ) u \|^2 ,
\end{equation}
for all $u \in \mathcal{H}$. We say that $0$ is a resonance of $H_V$ if the equation $H_V u = 0$ has a solution $u \in ( H^{1,s}( \mathbb{R}^3 ) \otimes \mathfrak{h} ) \setminus ( L^2 ( \mathbb{R}^3 ) \otimes \mathfrak{h} )$, for any $s > 1$, where $H^{1,s}( \mathbb{R}^3 )$ is the first-order Sobolev space on 
$\mathbb{R}^3$ with weight $\langle x \rangle^{-s}$. Applying Theorem \ref{thm:capture} we obtain the following result.
\begin{theorem}
Let $\mathcal{L}$ be given by \eqref{Lindbla2} and $\mathcal{L}_0 = \mathrm{ad}( - \Delta + H_{\mathrm{int}} )$.
Suppose that the conditions on $V$ described above are satisfied and that
\begin{equation*}
\big \|ÊC \langle X \rangle^{1+\varepsilon} \big \|Ê< 2 \mathrm{c}_1^{-1} < \infty,
\end{equation*}
for some $\varepsilon > 0$, where $\mathrm{c}_1$ is defined by \eqref{eq:EFG}. Then the modified wave operator $\tilde{\Omega}^-( \mathcal{L}_0 , \mathcal{L} )$ defined in \eqref{eq:modifwaveoperator} exists. 
\end{theorem}

\appendix

\section{Proof of Lemma \ref{21}} \label{app:21}

\begin{proof}
We sketch a proof, see also \cite[Lemma 5.1 and Theorem 5.2]{Daviesbook}. We only treat the case where $H_0$ is unbounded. We introduce the
  operator
  \begin{equation} \label{22} H:=H_0 - \frac{i}{2}\sum_{j\in J}
    C_j^{*}C_j
  \end{equation}
  on $\mathcal{H}$ with domain $\mathcal{D}(H_0)$. The dissipativity
  of $H$ is clear because
  \begin{equation*}
    \mathrm{Im} (\langle \varphi, H \varphi \rangle) = -\frac{1}{2} \sum_{j} \| C_j  \varphi \|_{ \mathcal{H} }^2 \leq 0. 
  \end{equation*}
  Furthermore, we claim that there exists $\lambda_0>0$ such that
  $H-i \lambda_0$ is bounded invertible, i.e.
  $(H-i \lambda_0)^{-1} \in \mathcal{B}(\mathcal{H})$.  Indeed, $H$ is
  closed because $H_0$ is self-adjoint and therefore, since in
  addition
  $\| (H-i \lambda_0 )\varphi \| \geq \lambda_0 \| \varphi \|$ for all
  $\varphi \in \mathcal{D}( H_0 )$ and all $\lambda_0>0$, we only have
  to show that the range of $H-i \lambda_0 $ is dense for some
  $\lambda_0>0$. This is equivalent to
  $\mathrm{Ker}(H^*+i \lambda_0)=\{0\}$. This last equality holds for
  any $\lambda_0 >0$ because
  \begin{equation*}
    H^* = H_0 + \frac{i}{2}\sum_{j\in J} C_j^{*}C_j,
  \end{equation*}
  and hence $H^* +i \lambda_0$ is injective.  The theorem of
  Lumer-Phillips (see e.g. \cite{Engel}) implies that the
  dissipative operator $H$ generates a strongly continuous
  one-parameter semigroup, $\{e^{-itH}\}_{t \geq 0}$ on $\mathcal{H}$.
  The linear operator $\mathrm{ad}(H)$ on
  $\mathcal{J}_{1}(\mathcal{H})$ with domain
  $\mathcal{D}(\mathrm{ad}(H_0))$ generates consequently a
  one-parameter semigroup of contractions given by
  \begin{equation}
    \rho \mapsto e^{-itH} \rho e^{itH^*}
  \end{equation}
  for all $\rho \in \mathcal{J}_{1}(\mathcal{H})$ and all $t \geq 0$.
  Here we use that
 $$\| e^{-itH} \rho e^{itH^*} \|_1 \le \| e^{-itH} \|_{ \mathcal{B}( \mathcal{H} ) } \|\rho \|_1 \|
  e^{itH^*} \|_{ \mathcal{B}( \mathcal{H} ) } \le \| \rho \|_1.$$
  This semigroup is clearly positivity preserving.  As the operator
  $$ i\sum_{j\in J}^{}C_j\, (\cdot) \,C_j^{*}$$ is bounded, a standard
  perturbation result for semigroups (see e.g. \cite{Engel}) shows
  that the operator $\mathcal{L}$ is defined and closed on
  $\mathcal{D}(\mathrm{ad}(H_0))$ and generates a strongly continuous
  one-parameter semigroup on $\mathcal{J}_{1}( \mathcal{H} )$. The
  semigroup $\{e^{-it \mathcal{L}} \}_{t \geq 0}$ satisfies
  \eqref{item:4} and \eqref{item:5}, i.e. it preserves positivity and
  the trace. Complete positivity follows from the Dyson series
  expansion of $e^{-it \mathcal{L}}$ (see \eqref{eq:DysonL1}--\eqref{eq:DysonL2}), using that
  $C_j\,(\cdot) \,C_j^{*}$ and $e^{-itH} (\cdot) e^{itH^*}$ are completely positive. Trace preservation is also clear by differentiating $t \mapsto \mathrm{tr}( e^{ - i t \mathcal{L} } \rho )$ for any $\rho \in \mathcal{D}( \mathcal{L} )$ and using that $\mathrm{tr}( \mathcal{L} \tilde{\rho} ) = 0$ for any $\tilde \rho \in \mathcal{D}( \mathcal{L} )$.

  Finally, the contractivity property of $e^{-it \mathcal{L}}$
  restricted to $\mathcal{J}_{1}^{\mathrm{sa}}(\mathcal{H})$ follows
  directly from the decomposition $\rho = \rho_+ + \rho_-$ with $\rho_+ = \rho \mathds{1}_{ [ 0 , \infty ) }( \rho )$, $\rho_- = \rho \mathds{1}_{ ( - \infty , 0 ] } ( \rho )$ and the fact that $\| \rho \|_1 = \mathrm{tr}( \rho_+ ) - \mathrm{tr}( \rho_- )$.
\end{proof}

\section{Appendix to Section \ref{sec:dissipative}}\label{app:dissipative}

In this appendix we use the notations of Section \ref{sec:dissipative}. We establish some properties of the wave operators $W_+( H_0 , H )$ and $W_-( H , H_0 )$. Some of them are already proven in \cite{Martin} and \cite{Davies1}. We give details for the sake of completeness.
\begin{lemma}\label{lm:limits}
Suppose that either Assumption \ref{Y0} or \ref{Z0} holds. Then 
\begin{equation}
\lim_{t \to \infty} \big \| W_+( H , H_0 ) e^{ i t H_0 } u \big \| = \| u \| , \label{eq:HH1}
\end{equation}
for all $u \in \mathcal{H}$. 
\end{lemma}
\begin{proof}
First suppose that Assumption \ref{Y0} holds. The existence of $W_+( H , H_0 )$ on $\mathcal{H}$ is a consequence of Cook's argument as recalled in Proposition \ref{prop:existenceCook}. In fact we have as in \eqref{eq:cook0} that
\begin{equation}\label{eq:conv_integral_cook}
W_+( H , H_0 ) u = u - \frac12 \int_0^\infty e^{Ê- i s H }ÊC^* C e^{ i s H_0 }Êu ds ,
\end{equation}
for all $u \in \mathcal{E}$. The integral in the right-hand side obviously converges by Assumption \ref{Y0}. Changing variables, we obtain from the previous identity that
\begin{align*}
ÊW_+( H , H_0 ) e^{ i t H_0 } u = Êe^{ i t H_0 } u - \frac12 \int_t^\infty e^{ - i ( s - t ) H } C^* C e^{ i s H_0 } u ds .
\end{align*}
Since Assumption \ref{Y0} holds,
\begin{equation*}
\Big \|Ê\int_t^\infty e^{ - i ( s - t ) H } C^* C e^{ i s H_0 } u ds \Big \|Ê\le \int_t^\infty \big \| C^* C e^{ i s H_0 } u \big \| ds \to 0 ,
\end{equation*}
as $t \to \infty$. Using the triangle inequality, this implies \eqref{eq:HH1} for all $u \in \mathcal{E}$. Using that $\mathcal{E}$ is dense in $\mathcal{H}$ we deduce that \eqref{eq:HH1} holds for all $u \in \mathcal{H}$.

Now suppose that Assumption \ref{Z0} holds. We can proceed in the same way. The existence of $W_+( H , H_0 )$ on $\mathcal{H}$ as well as the convergence of the integral in \eqref{eq:conv_integral_cook} are established in the proof of Proposition \ref{prop:existence}. Moreover, since Assumption \ref{Z0} holds, we can proceed as in the proof of Proposition \ref{prop:existence}, which gives
\begin{align*}
\Big \|Ê\int_t^\infty e^{ - i ( s - t ) H } C^* C e^{ i s H_0 } u ds \Big \|Ê&\le \sup_{ v \in \mathcal{H} , \| v \|Ê= 1Ê} \Big (Ê\int_0^\infty \big \| C e^{ i s H^* } v \big \|^2 ds \Big )^{\frac12} \Big ( \int_t^\infty \big \| C e^{ i s H_0 } u \big \|^2 ds \Big )^{\frac12}   \\
&\le \Big ( \int_t^\infty \big \| C e^{ i s H_0 } u \big \|^2 ds \Big )^{\frac12} \to 0 ,
\end{align*}
as $t \to \infty$. We then conclude, as above, that \eqref{eq:HH1} holds.
\end{proof}
\begin{proposition}\label{prop:injectivity}
Suppose that either Assumption \ref{Y0} or \ref{Z0} holds. Then $W_+( H , H_0 )$ is injective.
\end{proposition}
\begin{proof}
It suffices to combine Proposition \ref{prop:01} and Lemma \ref{lm:limits}. Indeed, suppose that $u \in \mathcal{H}$ satisfies $W_+( H , H_0 ) u = 0$. By Proposition \ref{prop:01},
\begin{equation*}
e^{ i t H } W_+ ( H , H_0 ) u = W_+( H , H_0 ) e^{ i t H_0 } u = 0 ,
\end{equation*}
for all $t \ge 0$. Letting $t \to \infty$ then shows that $u = 0$, by Lemma \ref{lm:limits}.
\end{proof}
 \begin{proposition}\label{prop:closedness}
Suppose that either Assumption \ref{Y0} or \ref{Z0} holds. Then $\mathrm{Ran}( W_+( H , H_0 ) )$ is closed if and only if the restriction of $\{ e^{ - i t H }Ê\}_{ t \in \mathbb{R} }$ to $\mathrm{Ran}( W_+( H , H_0 ) )$ is uniformly bounded.
\end{proposition}
\begin{proof}
First assume that $\{ e^{ - i t H }Ê\}_{ t \in \mathbb{R} }$ is uniformly bounded on $\mathrm{Ran}( W_+( H , H_0 ) )$. Let $M \ge 1$ be such that $\| e^{ - i t H } W_+( H , H_0 ) u \| \le M \|ÊW_+( H , H_0 ) u \|$ for all $t \in \mathbb{R}$ and $u \in \mathcal{H}$. Applying Lemma \ref{lm:limits} and Proposition \ref{prop:01} give
\begin{equation*}
\| u \| = \lim_{t \to \infty } \| W_+(H,H_0) e^{itH_0} u \|Ê= \lim_{t \to \infty } \| e^{itH} W_+(H,H_0) u \|Ê\le M \|ÊW_+( H , H_0 ) u \|,
\end{equation*}
for all $u \in \mathcal{H}$. Hence $W_+ ( H , H_0 )$ has closed range.

Suppose now that $\mathrm{Ran}( W_+( H , H_0 ) )$ is closed. Since $W_+ ( H , H_0 )$ is also injective by Proposition \ref{prop:injectivity}, there exists $m>0$ such that $\|W_+(H,H_0) u\| \ge m \|u\|$, for all $u\in\mathcal{H}$. Using  Proposition \ref{prop:01} and the fact that $W_+(H,H_0)$ is a contraction, this implies that
\begin{equation*}
\| e^{ - i t H } W_+( H , H_0 ) u \| = \| W_+( H , H_0 ) e^{ - i t H_0 }Êu \|  \le \|Êu \|Ê\le m^{-1}Ê\|ÊW_+(H,H_0) u \|Ê,
\end{equation*}
for all $t \in \mathbb{R}$ and $u \in \mathcal{H}$.
\end{proof}

 \section{Integral kernels and trace}\label{kernel}

 In order to study the wave operators on
 $\mathcal{J}_1\bigl( L^2( \mathbb{R}^3 \otimes \mathfrak{h} )\bigr)$ in Section \ref{sec:example},
 we exploited the integral kernel representation of Hilbert-Schmidt
 operators
 \begin{equation*}
   \mathcal{J}_2\bigl( L^2( \mathbb{R}^3 \otimes \mathfrak{h} ) \bigr)\supset \mathcal{J}_1\bigl( L^2( \mathbb{R}^3 \otimes \mathfrak{h} ) \bigr).
 \end{equation*}
In this appendix we provide some details about this representation. Let
 \begin{equation*}
   d := \mathrm{dim} \, \mathfrak{h} < \infty ,
 \end{equation*}
 and $\mathbb{Z}_d := \{ 1 , 2 , \dots , d \}$.  We recall the
 following well-known result (see e.g.
 \cite{MR541149}).

 \begin{proposition} \label{prop:6} We have the following isometric
   isomorphisms:
   \begin{equation*}
     \mathcal{J}_2\bigl( L^2( \mathbb{R}^3 \otimes \mathfrak{h} ) \bigr) \equiv L^2( \mathbb{R}^6 ; \mathfrak{h} \otimes \mathfrak{h} ) \equiv L^2 \big ( (\mathbb{R}^3 \times \mathbb{Z}_{d} )^2 \big ).
   \end{equation*}
 \end{proposition}

 Letting
 $i : \mathcal{J}_2\bigl( L^2( \mathbb{R}^3 \otimes \mathfrak{h} )
 \bigr) \to L^2 \big ( (\mathbb{R}^3 \times \mathbb{Z}_{d} )^2 \big )$
 be the isometric isomorphism of the proposition above,
 $L^2 \big ( (\mathbb{R}^3 \times \mathbb{Z}_{d} )^2 \big ) \ni a(
 \underline{x} , \underline{y} )=i(a)$
 is called the integral kernel of $a$, where
 $\underline{x} := ( x , \lambda )$ and $\underline{y} := ( y , \mu )$
 belong to $\mathbb{R}^3 \times \mathbb{Z}_d$. We will use the
 notation
 \begin{equation*}
   \int_{ \mathbb{R}^3 \times \mathbb{Z}_d } d \underline{x} := \sum_{ \lambda = 1 }^d \int_{ \mathbb{R}^3 } dx.
 \end{equation*}
 Let
 $\{\phi _j\}_j,\{\psi _j\}_j\subset L^2( \mathbb{R}^3 \times
 \mathbb{Z}_d )$
 be orthonormal collections; if
 $a=\sum_{j}\alpha _j\lvert\phi_j\rangle\langle\psi _j\rvert$, then
 $a ( \underline{x} , \underline{y} )=\sum_{j}\alpha
 _j\bar{\phi}_j(\underline{x})\psi _j(\underline{y})$,
 and the expansion converges absolutely a.e. (if the sum is infinite).

 We remark that by Proposition~\ref{prop:6}, every
 $a\in \mathcal{J}_1\bigl( L^2( \mathbb{R}^3 \otimes \mathfrak{h} )
 \bigr)$
 has an associated integral kernel $a(\underline{x},\underline{y})$;
 however, in general $a(\underline{x},\underline{x})$ may not be
 integrable. The following proposition gives a characterization of the
 trace by means of an Hardy-Littlewood averaging process
 $\mathcal{A}_n$ on $L^2( \mathbb{R}^3 \otimes \mathfrak{h} )$
 \cite{MR1123441}. Without entering too much into details, given a kernel
 $a(\underline{x},\underline{y})=\sum_{j}\alpha
 _j\bar{\phi}_j(\underline{x})\psi _j(\underline{y})$,
 let the kernel $\mathcal{A}_n^{(2)}a(\underline{x},\underline{y})$ be
 defined by
 \begin{equation*}
   \mathcal{A}_n^{(2)}a(\underline{x},\underline{y}) =\sum_{j}\alpha _j\mathcal{A}_n\bar{\phi}_j(\underline{x})\mathcal{A}_n\psi _j(\underline{y})\; .
 \end{equation*}
 Then the limit kernel $\tilde{a}(\underline{x},\underline{y})$ is defined as the pointwise
 a.e. limit
 \begin{equation*}
   \tilde{a}(\underline{x},\underline{y})=\lim_{n\to \infty } \mathcal{A}_n^{(2)}a(\underline{x},\underline{y}) \; .
 \end{equation*}

\begin{proposition}[\cite{MR1123441}]\label{prop:7}
Let
$a\in \mathcal{J}_1\bigl( L^2( \mathbb{R}^3 \otimes \mathfrak{h} )
\bigr)$,
with associated integral kernel $a(\underline{x},\underline{y})$. Then
the averaged kernel $\tilde{a}(\underline{x},\underline{x})$ exists
a.e., and
\begin{equation*}
  \mathrm{Tr}(a)=\int_{ \mathbb{R}^3 \times \mathbb{Z}_d } ^{}\tilde{a}(\underline{x},\underline{x})  d \underline{x} \; .
\end{equation*}
\end{proposition}
\begin{proposition}[\cite{MR1123441}]\label{prop:8}
  Let $a=bc$ be an arbitrary factorization of
  $k\in \mathcal{J}_1\bigl(L^2 ( \mathbb{R}^3 \otimes \mathfrak{h}
  )\bigr)$
  into a product of two Hilbert-Schmidt operators
  $b,c\in \mathcal{J}_2\bigl(L^2 ( \mathbb{R}^3 \times \mathbb{Z}_d )
  \bigr)$. Then
  \begin{equation*}
    \tilde{a}( \underline{x} , \underline{x} )=(b*c)( \underline{x} , \underline{x} )\text{ a.e.},
  \end{equation*}
  where the ``convoluted'' kernel
  $(b*c)(\underline{x} , \underline{y})$ is defined as
  \begin{equation*}
    (b*c) ( \underline{x} , \underline{y} ) = \int_{ \mathbb{R}^3 \times \mathbb{Z}_d } b( \underline{x} , \underline{z} ) c( \underline{z} , \underline{y} )  d\underline{z}\; .
  \end{equation*}
\end{proposition}
Proposition~\ref{prop:8} shows that, independently of the
factorization $a=bc$ of a trace class operator $a$, its trace is
always given by
$\int_{ \mathbb{R}^3 \times \mathbb{Z}_d } b( \underline{x} ,
\underline{z} ) c( \underline{z} , \underline{y} ) d\underline{z}$.
Since for any trace class operator there exist at least one such
decomposition, we may write the subspace of
$L^2( ( \mathbb{R}^3 \times \mathbb{Z}_d )^2 )$ corresponding to trace
class operators as
\begin{align*}
  \mathfrak{J}_1&:=\{a(\cdot ,\cdot )\in L^2 ( ( \mathbb{R}^3 \times \mathbb{Z}_d )^2 ), \exists b(\cdot ,\cdot ),c(\cdot ,\cdot ) \in L^2 ( ( \mathbb{R}^3 \times \mathbb{Z}_d )^2 ), a=b*c\} \\
  &\equiv \mathcal{J}_1\bigl(L^2 ( \mathbb{R}^3 \times \mathbb{Z}_d )\bigr)\; ;
\end{align*}
where the symbol $\equiv$ stands for an isometric isomorphism, and the
isometry is obtained defining the $\mathfrak{J}_1$ norm
\begin{equation*}
  \lVert a(\cdot ,\cdot )  \rVert_{\mathfrak{J}_1}^{}=\int_{ \mathbb{R}^3 \times \mathbb{Z}_d } ^{}\widetilde{\lvert a  \rvert_{}^{}}( \underline{x} , \underline{x} )  d \underline{x} \; .
\end{equation*}
Hence $(\mathfrak{J}_1,\lVert \cdot \rVert_{\mathfrak{J}_1})$ is a
Banach subspace of the Hilbert space
$L^2 ( ( \mathbb{R}^3 \times \mathbb{Z}_d )^2 )$.

\footnotesize

 \bibliography{FFFS_arxiv_february_12_2016}

\begin{thebibliography}{10}

\bibitem{Albert-Fr}
C.~Albert, L.~Ferrari, J.~Fr{\"o}hlich, and B.~Schlein.
\newblock Magnetism and the {W}eiss exchange field---a theoretical analysis
  motivated by recent experiments.
\newblock {\em J. Stat. Phys.}, 125(1):77--124, 2006.

\bibitem{Alicki1}
R.~Alicki.
\newblock On the scattering theory for quantum dynamical semigroups.
\newblock {\em Ann. Inst. H. Poincar\'e Sect. A (N.S.)}, 35(2):97--103, 1981.

\bibitem{AlickiFrigerio}
R.~Alicki and A.~Frigerio.
\newblock Scattering theory for quantum dynamical semigroups. {II}.
\newblock {\em Ann. Inst. H. Poincar\'e Sect. A (N.S.)}, 38(2):187--197, 1983.

\bibitem{AlickiLendi}
R.~Alicki and K.~Lendi.
\newblock {\em Quantum dynamical semigroups and applications}, volume 286 of
  {\em Lecture Notes in Physics}.
\newblock Springer-Verlag, Berlin, 1987.

\bibitem{BeMa92_01}
M.~Ben-Artzi and S.~Klainerman.
\newblock Decay and regularity for the {S}chr\"odinger equation.
\newblock {\em J. Anal. Math.}, 58:25--37, 1992.
\newblock Festschrift on the occasion of the 70th birthday of Shmuel Agmon.

\bibitem{Gol}
N.~Boussaid and S.~Gol\'{e}nia.
\newblock Limiting absorption principle for some long range perturbations of
  dirac systems at threshold energies.
\newblock {\em Commun. Math. Phys.}, 299(3):677--708, 2010.

\bibitem{MR1123441}
C.~Brislawn.
\newblock Traceable integral kernels on countably generated measure spaces.
\newblock {\em Pacific J. Math.}, 150(2):229--240, 1991.

\bibitem{Daviesbook}
E.~B. Davies.
\newblock {\em Quantum theory of open systems}.
\newblock Academic Press [Harcourt Brace Jovanovich, Publishers], London-New
  York, 1976.

\bibitem{davies77}
E.~B. Davies.
\newblock Quantum dynamical semigroups and the neutron diffusion equation.
\newblock {\em Rep. Mathematical Phys.}, 11(2):169--188, 1977.

\bibitem{Davies4}
E.~B. Davies.
\newblock Two-channel {H}amiltonians and the optical model of nuclear
  scattering.
\newblock {\em Ann. Inst. H. Poincar\'e Sect. A (N.S.)}, 29(4):395--413 (1979),
  1978.

\bibitem{davies79}
E.~B. Davies.
\newblock Generators of dynamical semigroups.
\newblock {\em J. Funct. Anal.}, 34(3):421--432, 1979.

\bibitem{Davies1}
E.~B. Davies.
\newblock Nonunitary scattering and capture. {I}. {H}ilbert space theory.
\newblock {\em Comm. Math. Phys.}, 71(3):277--288, 1980.

\bibitem{Davies2}
E.~B. Davies.
\newblock Nonunitary scattering and capture. {II}. {Q}uantum dynamical
  semigroup theory.
\newblock {\em Ann. Inst. H. Poincar\'e Sect. A (N.S.)}, 32(4):361--375, 1980.

\bibitem{Davies}
E.~B. Davies.
\newblock {\em Linear operators and their spectra}, volume 106 of {\em
  Cambridge Studies in Advanced Mathematics}.
\newblock Cambridge University Press, Cambridge, 2007.

\bibitem{Engel}
K.-J. Engel and R.~Nagel.
\newblock {\em One-parameter semigroups for linear evolution equations}, volume
  194 of {\em Graduate Texts in Mathematics}.
\newblock Springer-Verlag, New York, 2000.

\bibitem{Evans}
D.~E. Evans.
\newblock Smooth perturbations in non-reflexive {B}anach spaces.
\newblock {\em Math. Ann.}, 221(2):183--194, 1976.

\bibitem{IngardenKossakowski}
R.~S. Ingarden and A.~Kossakowski.
\newblock On the connection of nonequilibrium information theormodynamics with
  non-{H}amiltonian quantum mechanics of open systems.
\newblock {\em Ann. Physics}, 89:451--485, 1975.

\bibitem{JensenKato}
A.~Jensen and T.~Kato.
\newblock Spectral properties of {S}chr\"odinger operators and time-decay of
  the wave functions.
\newblock {\em Duke Math. J.}, 46(3):583--611, 1979.

\bibitem{Kadowaki}
M.~Kadowaki.
\newblock Resolvent estimates and scattering states for dissipative systems.
\newblock {\em Publ. Res. Inst. Math. Sci.}, 38(1):191--209, 2002.

\bibitem{Kato1}
T.~Kato.
\newblock Wave operators and similarity for some non-selfadjoint operators.
\newblock {\em Math. Ann.}, 162:258--279, 1965/1966.

\bibitem{KatoYajima}
T.~Kato and K.~Yajima.
\newblock Some examples of smooth operators and the associated smoothing
  effect.
\newblock {\em Rev. Math. Phys.}, 1(4):481--496, 1989.

\bibitem{kossa}
A.~Kossakowski.
\newblock On necessary and sufficient conditions for a generator of a quantum
  dynamical semi-group.
\newblock {\em Bull. Acad. Polon. Sci. S\'er. Sci. Math. Astronom. Phys.},
  20:1021--1025, 1972.

\bibitem{LiebLoss}
E.~H. Lieb and M.~Loss.
\newblock {\em Analysis}, volume~14 of {\em Graduate Studies in Mathematics}.
\newblock American Mathematical Society, Providence, RI, second edition, 2001.

\bibitem{Lin}
S.-c. Lin.
\newblock Wave operators and similarity for generators of semigroups in
  {B}anach spaces.
\newblock {\em Trans. Amer. Math. Soc.}, 139:469--494, 1969.

\bibitem{MR0413878}
G.~Lindblad.
\newblock On the generators of quantum dynamical semigroups.
\newblock {\em Comm. Math. Phys.}, 48(2):119--130, 1976.

\bibitem{Martin}
P.~A. Martin.
\newblock Scattering theory with dissipative interactions and time delay.
\newblock {\em Nuovo Cimento B (11)}, 30(2):217--238, 1975.

\bibitem{Mochizuki}
K.~Mochizuki.
\newblock Scattering theory for wave equations with dissipative terms.
\newblock {\em Publ. Res. Inst. Math. Sci.}, 12(2):383--390, 1976/77.

\bibitem{Phillips1}
R.~S. Phillips.
\newblock Perturbation theory for semi-groups of linear operators.
\newblock {\em Trans. Amer. Math. Soc.}, 74:199--221, 1953.

\bibitem{Rauch}
J.~Rauch.
\newblock Local decay of scattering solutions to {S}chr\"odinger's equation.
\newblock {\em Comm. Math. Phys.}, 61(2):149--168, 1978.

\bibitem{ReedSimon2}
M.~Reed and B.~Simon.
\newblock {\em Methods of modern mathematical physics. {II}. {F}ourier
  analysis, self-adjointness}.
\newblock Academic Press [Harcourt Brace Jovanovich, Publishers], New
  York-London, 1975.

\bibitem{RS-III}
M.~Reed and B.~Simon.
\newblock {\em Methods of modern mathematical physics. {III}. Scattering
  theory}.
\newblock Academic Press [Harcourt Brace Jovanovich, Publishers], New
  York-London, 1979.

\bibitem{reuteler}
J.~Reuteler.
\newblock Qm particles scattering off a dynamical target.
\newblock {\em ETHZ diploma thesis}, 216(2):303--361, 2004.

\bibitem{Roy}
J.~Royer.
\newblock Limiting absorption principle for the dissipative {H}elmholtz
  equation.
\newblock {\em Comm. Partial Differential Equations}, 35(8):1458--1489, 2010.

\bibitem{Simon2}
B.~Simon.
\newblock Phase space analysis of simple scattering systems: extensions of some
  work of {E}nss.
\newblock {\em Duke Math. J.}, 46(1):119--168, 1979.

\bibitem{MR541149}
B.~Simon.
\newblock {\em Trace ideals and their applications}, volume~35 of {\em London
  Mathematical Society Lecture Note Series}.
\newblock Cambridge University Press, Cambridge-New York, 1979.

\bibitem{Simon}
B.~Simon.
\newblock Best constants in some operator smoothness estimates.
\newblock {\em J. Funct. Anal.}, 107(1):66--71, 1992.

\bibitem{WangZhu}
X.~P. Wang and L.~Zhu.
\newblock On the wave operator for dissipative potentials with small imaginary
  part.
\newblock {\em Asymptot. Anal.}, 86(1):49--57, 2014.

\bibitem{Sieg}
W.~Weber, S.~Riesen, and H.~C. Siegmann.
\newblock Magnetization precession by hot spin injection.
\newblock {\em Science.}, 291(5506):1015--1018, 2001.

\end{thebibliography}

\end{document}